\documentclass{article}

\usepackage[a4paper,margin=2.5cm]{geometry}

\usepackage{amsmath,amssymb,amsfonts}
\usepackage{amsthm}
\usepackage{physics}
\usepackage{braket}
\usepackage{bbm}
\usepackage{mathtools}
\usepackage{microtype}
\usepackage[acronym]{glossaries}

\usepackage{graphicx}
\usepackage[caption=false,font=footnotesize]{subfig}

\usepackage{multirow}
\usepackage{booktabs}
\usepackage{array}
\usepackage{enumitem}

\usepackage{xcolor}
\usepackage{comment}
\usepackage{marvosym}
\usepackage[english]{babel}
\usepackage[colorlinks=true,linkcolor=blue,citecolor=blue,urlcolor=blue]{hyperref}

\newtheorem{theorem}{Theorem}

\newtheorem{lemma}[theorem]{Lemma}

\newacronym{sbs}{SBS}{Stimulated Brillouin Scattering}
\newacronym{qkd}{QKD}{Quantum Key Distribution}
\newacronym{cco}{CCO}{Collective Casimir Operator}
\newacronym{povm}{POVM}{Positive Operator-Valued Measurement}
\newacronym{pvm}{PVM}{Projection-Valued Measurement}
\newacronym{pun}{PUN}{Passive Unitary Normalizable}
\newacronym{sld}{SLD}{Symmetric Logarithmic Derivative}
\newacronym{qfi}{QFI}{Quantum Fisher Information}
\newacronym{cr}{CR}{Cramer Rao}
\newacronym{bss}{BSS}{Blind Source Separation}
\newacronym{tdma}{TDMA}{Time Division Multiple Access}
\newcommand{\eins}{\mathbbm{1}}

\newcommand{\bb}{\mathbf{b}}
\newcommand{\ba}{\mathbf{a}}
\newcommand{\cov}{\mathrm{cov}}
\newcommand{\svar}{\mathrm{var}}

\newcommand{\clip}{\mathop{\text{\LeftScissors}}\nolimits}

\DeclareMathOperator{\off}{off}

\DeclareMathOperator{\diag}{diag}

\allowdisplaybreaks

\title{
Quantum-Limited Blind Source Separation of Classical Light
}

\author{
Janis Nötzel, Kiran Adhikari\\[0.5em]
\small Emmy Noether Group, Theoretical Quantum Systems Design\\
\small Technical University of Munich, Munich, Germany
}

\date{}
\begin{document}

\maketitle

\begin{abstract}
Using the framework of quantum multiparameter estimation, we study the problem of separating independent thermal optical sources mixed by an unknown passive linear transformation, which is known as blind source separation in signal processing. We propose a sensing-assisted method that iteratively estimates and suppresses optical correlations directly within the unit cells of a programmable interferometer. We show that collective quantum measurements exhibit a substantial advantage over conventional detection methods by constructing collective measurements that asymptotically achieve the Holevo Cramér--Rao bound in a unit cell of the interferometer for the respective sub-problem. By systematically arranging multiple such cells in a photonic mesh, our proposal implements an in situ Jacobi diagonalization of the input multimode correlation matrix. We contrast our method with a conventional approach where heterodyne detection is used to reconstruct the full covariance matrix and subsequently diagonalize it on a classical computer. A comparison with this heterodyne tomography approach shows that our sensing-based method is particularly advantageous in the weak-light regime.
\end{abstract}

\section{Introduction}
    We study the problem of blind source separation from a quantum perspective, where $K$ thermal light sources $\sigma =S_{N_1}\otimes\ldots\otimes S_{N_K}$ are mixed by a passive linear optical channel described by a unitary $U$. The resulting state $\rho=US_{N_1}\otimes\ldots\otimes S_{N_K}U^\dagger$ is the output visible to an analyzer attempting to identify the individual energy levels $N_1,\ldots,N_K$ as well as the mixing process $U$. By singling out first the base case of two modes $U(\theta,\phi)(S_M\otimes S_N)U(\theta,\phi)^\dagger$ we prove that the Holevo \gls{cr} bound for joint measurement of $(M,N,\theta,\phi)$ and gain matrix $G=\mathrm{diag}(1,1,1,\sin^2(2\theta)/4)$ is achievable and describe the asymptotically optimal \glspl{povm}. Our description includes sharp \gls{cr}-optimal measurements of $(M,N,\theta)$ when $\phi=0$. 
    
    We then move on to consider a hypothetical Holevo-optimal implementation of the measurement procedure in the sense of an elementary cell with two quantum optical in- and outputs as well as a classical input for starting and stopping the measurement process, and a classical output for the acquired classical information $(M,N,\theta,\phi)$. We describe the use of such a cell for the process of estimating a multi-parameter representation $U(\hat\Theta)$ of $U$ when each cell is embedded and connected in a neural network structure, where the overall network learns the structure of $U$ and thereby allows an observer to eventually obtain the parameters $N_1,\ldots,N_K$. We provide a comparison of this method against the state of the art, where heterodyne detection is used to obtain an estimate $\hat V$ of the true covariance matrix $V$ of $\rho$ and the model parameters are derived from $\hat V$. 

    Conceptually, the transformation $\rho\to \sigma = U(\hat\theta,\hat\phi)^\dagger\rho U(\hat\theta,\hat\phi)$ is similar to the diagonalization of a matrix. On the level of transformations of the covariance- and correlation matrices of $\rho$ and $\sigma$, this correspondence is exact. Our analysis shows that, in the regime of low photon numbers $N_1,\ldots,N_K\ll1$, the proposed quantum sensing- based Jacobi method outperforms aforementioned conventional heterodyne detection-based method. Not surprisingly, the respective quantum advantage is inversely proportional to the thermal energies $N_1,\ldots,N_K$.

    Beyond blind source separation, our proposed architecture provides a complementary perspective on quantum computation for signal processing. Rather than encoding a classical problem into a universal discrete quantum circuit, the optical field carrying the signal is processed directly by quantum-optical transformations and measurements, while the computation emerges from the physical evolution and adaptive configuration of the device. This distinction is relevant in view of recent (un-) computability results showing fundamental limitations of universal gate-based quantum computing, including quantum gate-circuit emulation, concatenation, and general quantum compiling \cite{BocheBoeckEtAl2025,BoeckBocheEtAl2026}. The present approach instead exploits a task-specific quantum processor acting directly on the physical signal, here a classical thermal optical field, and thereby points toward quantum-enabled analog signal processing as an alternative computational paradigm.

    Most importantly, the proposed architecture clarifies based on a particular problem instance the advantage of avoiding a process, by which data is first \emph{made} digital before being re-encoded into a (possibly quantum) computer for analysis.

    \subsection{Related Work}
    \subsubsection{Independent Component Analysis}
        The well-established problem of \gls{bss} provides a closely related perspective on the recovery of latent structure from multichannel observations. In its standard form, one assumes that the measured signals are unknown mixtures of a collection of latent sources. For an linear mixing model, the observations may be written as
        \begin{equation}
            x^K(t)
            =
            As^K(t)
            +
            \nu^K(t),
            \label{eq:bss-mixing-model}
        \end{equation}
        where
        \begin{equation}
            s^K(t) = \left(s_1(t), \ldots, s_K(t) \right)^{\mathsf{T}}
        \end{equation}
        contains the unknown source signals, $A$ is an unknown mixing matrix, and $\boldsymbol{\nu}(t)$ describes measurement noise. The objective is to construct an unmixing matrix $W$ such that
        \begin{equation}
            y^K(t) = Wx^K(t) \approx s^K(t) + \nu'^K(t),
            \label{eq:bss-unmixing-model}
        \end{equation}
        despite the absence of direct knowledge of either $A$ or $\mathbf{s}(t)$ at the computing- or signal processing unit tasked with the inversion ~\cite{jutten1991blind,comon1994independent, hyvarinen2000independent}.
        
        The problem cannot be solved without added assumptions. Independent component analysis, in particular, assumes that there are $K$ sources which are statistically independent and typically exploits their non-Gaussianity. Separation can then be
        formulated as the minimization of statistical dependence, the maximization of non-Gaussianity, the diagonalization of higher-order cumulant matrices, or the maximization of information transferred through an adaptive nonlinear network~\cite{comon1994independent, cardoso1993blind,bell1995information}. Other formulations use temporal
        correlations, nonstationarity, spectral diversity, sparsity, or prior information about the mixing process. These assumptions are essential: without sufficient source structure, the mixing transformation and the
        latent signals are generally not identifiable.
        
        Even under added assumptions, the recovered sources are typically determined only up to ambiguities such as e.g. permutations among the $K$ different sources. 
        
        The optical correlation problem considered by us may also be interpreted as a structured source-separation problem. A field generated by a number $K$ of independent thermal sources is received by several detectors after mixing through a passive linear unitary map. The goal of the receiver is to decompose the source into the source- and mixing part. 
        

    \subsubsection{Multi-parameter estimation and Quantum Fisher Information}\label{subsec:quantum-fisher-metric}
        
        A parametrized family of quantum states
        \begin{equation}
            \mathcal{S} =\{
                \rho_{\Theta} : \Theta
                =
                (\theta_1,\ldots,\theta_p)
                \in\mathcal{R}
            \}
        \end{equation}
       where $\mathcal R\subset\mathbb R^p$ may be regarded as a statistical manifold. Each parameter direction $\partial_i \equiv \partial/\partial\theta_i$ is associated with a symmetric logarithmic derivative $L_i$ \cite{monras2013phase}, defined implicitly by
        \begin{equation}
            \partial_i\rho_{\Theta}
            =
            \frac{1}{2}
            \left(
                L_i\rho_{\Theta}
                +
                \rho_{\Theta}L_i
            \right).
            \label{eq:sld-multiparameter}
        \end{equation}
        The \gls{sld} quantum Fisher information matrix is
        \begin{equation}
            \left(F_Q(\Theta)\right)_{ij} = \frac{1}{2} \Tr(\rho_{\Theta}\left\{L_i,L_j\right\}),
            \label{eq:qfim-definition}
        \end{equation}
        where $\{A,B\}=AB+BA$ denotes the anticommutator. We use the convention in which $F_Q$ itself is called the quantum Fisher metric.
        The quantum Fisher metric measures the infinitesimal statistical distinguishability of neighboring quantum states. It also has an operational interpretation: it is obtained by maximizing the classical Fisher information over all quantum measurements~\cite{braunstein1994statistical}.
        
        For a positive-operator-valued measure $\mathcal{M}=\{M_x\}_x$, the outcome distribution is
        \begin{equation}
            p(x\mid\Theta) = \Tr(\rho_{\Theta}M_x).
        \end{equation}
        The associated classical Fisher information matrix is
        \begin{equation}
            \left(F_{\mathsf{M}}(\Theta)\right)_{ij} = \sum_x p(x\mid\Theta)[\partial_i\log p(x\mid\Theta)][\partial_j\log p(x\mid\Theta)],
            \label{eq:classical-fisher-matrix}
        \end{equation}
        with the sum replaced by an integral for continuous outcomes. For every measurement,
        \begin{equation}
            F_{\mathsf{M}}(\Theta) \leq F_Q(\Theta),
            \label{eq:measurement-qfi-inequality}
        \end{equation}
        where $\leq$ denotes the positive-semidefinite matrix ordering in above equation. The quantum Fisher metric consequently provides a measurement-independent upper bound on the information obtainable from the state. For $k$ independent copies and a locally unbiased estimator $\widehat{\Theta}$, the \gls{sld} quantum \gls{cr} bound gives
        \begin{equation}
            \cov(\hat{\Theta}) \geq \frac{1}{k} F_Q(\Theta)^{-1}.
            \label{eq:multiparameter-qcrb}
        \end{equation}
        Equivalently, for a positive cost matrix $G$,
        \begin{equation}
            \Tr(G\cdot\cov(\hat\Theta))\geq\frac{1}{k}\Tr(G\cdot F_Q(\Theta)^{-1}).
            \label{eq:weighted-sld-bound}
        \end{equation}
        In a single-parameter problem this bound is asymptotically attainable under standard regularity assumptions. In a multiparameter problem, however, the measurements that are optimal for different parameter directions may be incompatible. Equation~\eqref{eq:multiparameter-qcrb} is then a valid lower bound in every direction but need not be
        simultaneously attainable by a single measurement.
        
        A tighter operational benchmark is provided by the Holevo Cram\'er--Rao bound~\cite{holevo2011probabilistic,
        ragy2016compatibility}. Let $\boldsymbol{X}=(X_1,\ldots,X_p)$ be a collection of Hermitian
        operators satisfying the local-unbiasedness conditions 
        \begin{align}
            \operatorname{Tr}
            \left[
                \rho_{\Theta}X_i
            \right]
            &=
            0,
            \label{eq:holevo-unbiasedness-mean}
            \\
            \operatorname{Tr}
            \left[
                \left(
                    \partial_j\rho_{\Theta}
                \right)X_i
            \right]
            &=
            \delta_{ij}.
            \label{eq:holevo-unbiasedness-derivative}
        \end{align}
        For the (generally complex-valued) matrix $Z_{ij}[\boldsymbol{X}] = \Tr(\rho_{\Theta}X_iX_j)$, the Holevo cost is \cite{ragy2016compatibility}
        \begin{equation}
            C_{\mathrm{H}}\left(G,\Theta\right) = 
            \min_{\boldsymbol{X}}
            \left\{
                \Tr(G\,\Re\{Z[\boldsymbol{X}]\})
                +
                \left\|
                    \sqrt{G}\,
                    \Im\{Z[\boldsymbol{X}]\}
                    \sqrt{G}
                \right\|_1
            \right\},
            \label{eq:holevo-cost}
        \end{equation}
        where $\|\cdot\|_1$ denotes the trace norm. The corresponding bound for $k$ (potentially collective) estimation steps is
        \begin{equation}
            \Tr(G\cdot \cov(\hat\Theta)) \geq \frac{1}{k}C_{\mathrm{H}}\left(G,\Theta\right).
            \label{eq:holevo-bound}
        \end{equation}
        The real part in Eq.~\eqref{eq:holevo-cost} describes the ordinary estimation covariance, while the trace-norm term penalizes the incompatibility of the observables required to estimate different parameters. Consequently,      \begin{equation}
            C_{\mathrm{H}}\left(G,\Theta\right)\geq\Tr(G\cdot F_Q(\Theta)^{-1}).
            \label{eq:holevo-stronger-than-sld}
        \end{equation}
        The two expressions agree in a single-parameter problem, but they may differ when several parameters must be inferred simultaneously. An important compatibility condition is
        \begin{equation}
            \Tr(\rho_{\Theta}[L_i,L_j]) = 0\qquad \forall\ i,j,
            \label{eq:weak-commutativity}
        \end{equation}
        where $[A,B]=AB-BA$ is the usual commutator. This condition is weaker than requiring the \glspl{sld} to commute as operators. Under the regularity and positive-definiteness assumptions considered in \cite{ragy2016compatibility}, \ref{eq:weak-commutativity} is necessary and sufficient for the Holevo bound to coincide with the quantum Fisher information bound in the asymptotic collective-measurement setting. 

    \subsubsection{Quantum-Optimal Optical Processing and Learning}
        A related line of work has investigated how quantum-information methods can improve the extraction and processing of classical information carried by optical fields. An early example is \cite{Guha2011Structured}, which showed that structured joint-detection receivers can achieve superadditive communication rates and approach the Holevo limit by coherently processing multiple optical symbols before detection.
        In \cite{Rosati2024LearningTheory}, a statistical learning theory for continuous-variable photonic circuits, establishing learnability guarantees for states, measurements, and channels generated by such processors, was proposed. Subsequently, the problem of learning a quantum receiver was considered in \cite{RosatiSolana2024}. The learning of quantum processes whose governing classical parameters cannot be actively controlled by the observer was the subject of study in \cite{FanizzaQuekRosati2024}. Recent work on photon-starved polarimetry to reconstruct properties of weak classical optical fields from limited photon counts was carried out in \cite{RosatiEtAl2026Polarimetry}. 
        
        In passive imaging, the work \cite{GraceEtAl2020} demonstrated that mode-selective optical processing can substantially outperform conventional direct imaging even when the required receiver basis is initially unknown, an adaptive modal receiver attaining the quantum Fisher information for localization of incoherent sources was developed in \cite{SajjadEtAl2021}. Object-independent linear-optical processing was used for discrimination between incoherent optical objects in \cite{GraceGuha2022}. 
        
        Together, these results demonstrate that optical fields carrying classical information can contain information that is not optimally accessed by conventional detection, while appropriately designed quantum-optical techniques provide strictly improved performance.

        In this context, the present work takes a complementary step by treating an unknown mode basis itself as the object to be learned and physically inverted: Quantum-limited local estimation is embedded into a self-configuring passive optical network that iteratively diagonalizes an unknown multimode thermal field, such that all inference and signal processing are performed directly on the incoming optical state.

    \subsubsection{Butterfly Networks}\label{subsubsec:butterfly-networks}
        
        Butterfly networks are sparse, layered architectures for implementing structured linear transformations. Their name derives from the
        characteristic pattern of pairwise interactions appearing in divide-and-conquer algorithms such as the fast Fourier transform. For
        $N=2^L$ inputs, a butterfly transform contains $L=\log_2 N$ stages, with each stage coupling the modes in $N/2$ disjoint pairs. 
        
        Butterfly networks have recently been studied as parts of machine-learning models  \cite{dao2019butterfly,ailon2021sparse,lin2021deformable}. In order to increase their expressivity, butterfly networks have been combined with other layers introducing e.g. permutations and further operations \cite{dao2019butterfly,ailon2021sparse}. 
        
        \subsubsection{Relation to linear-optical networks}
        
        A lossless $K$-mode linear-optical circuit can be described completely by a matrix $U\in U(K)$. It transforms $K$-mode coherent states $|\alpha^K\rangle = |\alpha_1\rangle\otimes\ldots\otimes|\alpha_K\rangle$ according to 
        \begin{equation}
            |\alpha^K\rangle_{\mathrm{in}} \to |U\alpha^K\rangle_{\mathrm{out}}.
            \label{eq:optical-unitary}
        \end{equation}
        The same matrix can be used to describe the transformation of bosonic creation- and annihilation operators. The elementary component (also called ``unit cell'', a name we shall borrow later) of many programmable photonic processors is a tunable Mach--Zehnder interferometer. Together with phase shifters, it implements (up to equivalent phase conventions) a two-mode transformation
        \begin{equation}
            U_{pq} = \eins^{\otimes(p-1)} \otimes V_p(\phi)U_{p,p+1}(\theta) \otimes \eins^{\otimes(K-p-1)},
            \label{eq:optical-two-mode-rotation}
        \end{equation}
        embedded into a total of $K$ modes. Here, $V(_p\phi)$ denotes a phase shift of the $p$-th mode and $U_{pq}$ a beam-splitter $U(\theta)=\exp(\theta(a^\dagger_pa_q-a_pa^\dagger_q))$. Thus, the central optical building block is precisely a complex two-dimensional rotation.
        
        Universal interferometers arrange these two-mode transformations so that an arbitrary element of $U(K)$ can be synthesized. The triangular construction of \cite{reck1994experimental} and the rectangular construction of \cite{clements2016optimal} require
        \begin{equation}
            \frac{K(K-1)}{2}
        \end{equation}
        tunable two-mode couplers, in addition to single-mode phase. The quadratic number of elements reflects the $K^2$ real dimensions of
        $U(K)$.
        
        The universal processor demonstrated in \cite{taballione2019reconfigurable,taballione2021universal} is of
        this general type. While an optical processor such as e.g. \cite{taballione2021universal} is universal and can implement any passive linear-optical transformation on its input modes, an optical butterfly network will for example, efficiently
        implement only the Fourier transform, Hadamard-type transforms, and learned hierarchical mode mixers.
        
        To implement a general nonunitary neural-network weight matrix $W\in\mathbb{C}^{M\times N}$, one may use the singular-value
        decomposition
        \begin{equation}
            W
            =
            U\Sigma V^\dagger.
            \label{eq:optical-svd}
        \end{equation}
        In conventional optical neural networks, $U$ and $V$ are implemented
        using universal interferometer meshes, while $\Sigma$ is implemented
        using mode-dependent attenuation or amplification
        ~\cite{shen2017deep}.

        \subsubsection{Trainable Complete Optical Meshes}
            Programmable optical interferometer meshes have been developed as trainable hardware for implementing linear and unitary transformations. As an example, the work \cite{Pai2019MatrixOptimization} established optimization methods for configuring universal photonic circuits to realize target unitary matrices and analyzed their robustness to device imperfections. An efficient training method that preserves unitarity throughout optimization was introduced in \cite{Lu2023EfficientTraining}. An experimental demonstration of a coherent nanophotonic neural network based on a programmable Mach–Zehnder-interferometer mesh applied it to classification tasks was carried out in \cite{shen2017deep}. The proposed method aimed at displaying a computational speed enhancement of at least two orders of magnitude over the state-of-the-art and three orders of magnitude in power efficiency for conventional learning tasks. The authors of \cite{Hughes2018InSituBackpropagation} showed based on classical optics that gradients of photonic-network parameters can be measured directly in the physical device through in-situ optical backpropagation, enabling hardware-aware training . In \cite{Miller2013SelfConfiguring}, a self-configuring universal linear optical network was proposed that can learn desired input–output mode transformations using local feedback rather than global numerical optimization . Finally the work \cite{Annoni2017Unscrambling} experimentally realized a self-configuring photonic circuit that automatically learned to reverse strong mixing between optical modes, demonstrating adaptive mode unscrambling.

        \subsubsection{Relation to the Jacobi algorithm}
        
        The connection between current passive optical networks and the Jacobi algorithm arises from the concatenated use of two-dimensional plane rotations, which is the core ingredient to both. Consider a Hermitian matrix $A=A^\dagger$. A Jacobi step selects a pair of indices $(p,q)$ and constructs a rotation $G_{pq}$ that acts nontrivially
        only in the corresponding two-dimensional subspace. The matrix is updated according to
        \begin{equation}
            A^{(k+1)} = G_{pq}^{\dagger} A^{(k)} G_{pq},
            \label{eq:jacobi-update}
        \end{equation}
        with the rotation angle and phase chosen so that
        \begin{equation}
            \left[A^{(k+1)}\right]_{pq} = \left[A^{(k+1)}\right]_{qp} = 0.
            \label{eqn:jacobi-annihilation}
        \end{equation}
        Repeated sweeps over index pairs drive the off-diagonal norm toward zero, yielding
        \begin{equation}
            Q^\dagger A Q \approx \Lambda, \qquad Q=\prod_k G_{p_kq_k},
            \label{eqn:jacobi-eigendecomposition}
        \end{equation}
        where $\Lambda$ is diagonal. This method is a standard tool for eigenvalue, singular-value, and matrix-factorization
        problems \cite{givens1958computation,golub2000eigenvalue,golub2013matrix}.
        
        Mathematically, each Jacobi rotation has the same form as a tunable two-mode optical interferometer. A sequence of Jacobi rotations can therefore be interpreted as a sequence of Mach--Zehnder settings, and the accumulated eigenvector matrix $Q$ can be implemented directly as a linear-optical network.
        
        Universal optical decomposition algorithms are closely related to this procedure but are not identical to the classical symmetric Jacobi algorithm. Reck- and Clements-type algorithms use \emph{one-sided} Givens eliminations to zero selected entries of a target \emph{unitary} and thereby factor it into optical components. Jacobi diagonalization instead uses the \emph{two-sided} similarity transformation in Eq.~\eqref{eq:jacobi-update}, so that Hermiticity and the eigenvalues of $A$ are preserved. Nevertheless, both methods reduce a global matrix problem to a sequence of programmable $2\times2$
        operations.
        
        This connection is especially explicit in self-configuration algorithms for rectangular photonic meshes. A rectangular interferometer was programmed using left and right products of Givens rotations that diagonalize a matrix associated with
        the target transformation in \cite{hamerly2022accurate}. The procedure is Jacobi-like in that successive local rotations annihilate selected matrix elements, while the optical hardware physically realizes each rotation.
        
        A butterfly network may in turn be interpreted as a predetermined, parallel schedule of Givens rotations. Summing up, a Jacobi algorithm chooses off-diagonal matrix entries (so-called ``pivots'') as well as rotation angles and $2\times2$ unitary transformations, which set the pivots to zero. A sequence of pivots which addresses every off-diagonal entry once is called a ``sweep''. Usually, a larger number of sweeps is required before a matrix is diagonalized. A butterfly network fixes a sparse sequence of pairings in advance and normally performs only logarithmically many stages. 
                
        In machine learning, a product of selected Givens rotations provides a compact parametrization of orthogonal or unitary weight matrices~\cite{frerix2019approximating}. In photonics, the same parametrization maps directly to interferometer settings. A butterfly architecture further groups these rotations into shallow parallel layers. Thus, butterfly networks, universal optical meshes, and Jacobi methods are all constructed from two-mode rotations. They differ primarily in which pairs are coupled, how the rotations are scheduled, whether the schedule is fixed or adaptive, and whether the resulting architecture is universal or restricted.

        In this work in particular, we ask the question how the $2\times2$ rotations could be implemented in an optimal way, where optimality is defined based on the notion of quantum parameter estimation. 

\subsection{Conventional Method}
    The conventional way of estimating the parameters $(M,N,\theta,\phi)$ from $k$ copies of $\rho$ is to apply $2$ heterodyne measurements on both A and B mode. Successifely, the measurement results are used to estimate the covariance matrix in the Gaussian state formalism. Most importantly, this object is then assumed to be accessible to standard computational tools. As a result, the conventional Jacobi method may be used to diagonalize it. In a next step, the parameters $\theta,\phi$ can be obtained by a standard decomposition of the unitary diagonalizing the covariance matrix. Due to the one-to-one relationship between covariance- and correlation matrices for the \gls{pun} states under consideration, this is completely sufficient for diagonalizing the correlation matrix as well, end delivers in particular $(M,N,\theta,\phi)$.

\subsection{Quantum Method}
    We first explain the functionality of an individual cell which is able to diagonalize one $2\times2$ block of the correlation matrix $\mathbf A=[\Tr(a_i^\dagger a_j\rho)]_{ij}$. The method is built from the fundamental problem of estimating $(\theta,\phi)$ in the model $\rho=V_A(\phi)U_{AB}(\theta)[S_M\otimes S_N] U_{AB}(\theta)^\dagger V_A(\phi)^\dagger$, which we describe in what follows. Afterwards, we detail how a layered architecture based on several of the cells can be used to implement the transformation $\rho\to\sigma$.
    \subsubsection{Unit Cell}\label{subsubsec:unit cell}
        To start with, we assume $\phi=0$. We define three commuting collective $k$-mode measurements and resulting combined estimators, which together let us approach the \gls{qfi} \gls{cr} bound \cite{holevo2011probabilistic,Helstrom1969}. These are built from non-destructive photon- and photon difference counting plus one collective spin measurement. To define the measurements, we introduce the corresponding \glspl{povm} and operators:
        The nondestructive measurement of the total photon number in mode $X\in\{A^k,B^k\}$ is based on the \gls{povm} $\{P_t^X\}_{t\in\mathbb N}$, where
        \begin{align}
            P_t^X:=\sum_{n^k:\sum_in_i=t}|n_1\rangle\langle n_1|\otimes\ldots\otimes|n_k\rangle\langle n_k|.
        \end{align}
        The corresponding measured photon numbers are given by $t\in\mathbb N$, and the operator describing the measurement when carried out on mode $X\in\{A,B\}$ is $N_X:=\sum_t t\cdot P_t^X$. If only the total photon number of both modes together is of interest, the corresponding \gls{povm} is given by 
        \begin{align}
            P_t^{AB}:=\sum_{r+s=t}P_r^A\otimes P_s^B
        \end{align}
        and the associated operator is $N_{AB}:=\sum_tt\cdot P_t^{AB}$. To measure photon number differences in a nondestructive way we use the \gls{povm} defined via operators 
        \begin{align}
            \mathrm{PND}_t^{AB}:=\sum_{r-s=t}P_r^A\otimes P_s^B,
        \end{align}
        which define the operator $J_z:=\tfrac{1}{2}\sum_tt\cdot \mathrm{PND}_t^{AB}$. Equivalently, using the creation- and annihilation operators $a_i^\dagger$, $b_i^\dagger$ and $a_i$, $b_i$ on the $A$ (and $B$) part of the $i$-th copy of our system we can write $J_z$ as 
        \begin{align}
            J_z:=\tfrac{1}{2}\sum_i(a_i^\dagger a_i - b_i^\dagger b_i).
        \end{align}
        Upon introducing in addition corresponding collective Schwinger operators
        \begin{align}
            J_+:=\sum_ia_i^\dagger b_i,\qquad J_-:=\sum_ib_i^\dagger a_i,
        \end{align}
        we can construct a (collective) Casimir operator \cite{Casimir1931}
        \begin{align}
            \mathcal{J}&:=J_z^2 + \tfrac{1}{2}(J_+J_-+J_-J_+)\\
                &=J_x^2+J_y^2+J_z^2
        \end{align}
        where $J_x:=(J_++J_-)/2$ , $J_y:=\mathbbm{i}(J_--J_+)/2$. It then follows that 
        \begin{align}\label{eqn:commutators}
            [\mathcal{J},J_z]&=[\mathcal{J},J_+]=[\mathcal{J},J_-]=[\mathcal{J},P_s\otimes P_t]=[\mathcal{J},U(\theta)^{\otimes k}]=0\\
            [J_-,J_+]&=-2J_z.
        \end{align}
        Thus, $\mathcal{J}$ can be measured jointly with the photon numbers $s$ and $t$ in modes A and B, as well as with joint photon numbers and photon number differences. For the measurement, the operator that is utilized is however not $\mathcal{J}$ itself, but (with some deviation from the standard literature, where the label $J$ would be used to denote $\mathcal{J}$) rather 
        \begin{align}
            J = \sqrt{\mathcal{J}+\tfrac{1}{4}}-\frac{1}{2}.
        \end{align}
        The estimate achieving the Holevo Cramer-Rao bound is then based on measuring $J=\sqrt{\mathcal{J}+1/4}-1/2$ together with collective photon number $N_{AB}$ of both modes $(A,B)$ and the photon number difference $J_z=(N_A-N_B)/2$ to obtain an estimate for $(M,N,\tau)$ as
        \begin{align}
            \hat M = (\tfrac{1}{2}N_{AB}+J)/k,\qquad\hat N = (\tfrac{1}{2}N_{AB}-J)/k,\qquad\hat\tau = \tfrac{1}{2}+J_z/2J.
        \end{align}
        For the case of arbitrary $(\theta,\phi)$ we define $J_x:=(J_++J_-)/2$ , $J_y:=\mathbbm{i}(J_--J_+)/2$. Instead of measuring $J_z$, the third measurement is now given by the \gls{povm} elements
        \begin{align}\label{def:omega-povm}
            M_{j,t}(\omega):=E_t\frac{2j+1}{4\pi} P_{j,\omega}E_t,\qquad\omega\in S^2,
        \end{align}
        where $J_\omega|j,\omega\rangle = j|j,\omega\rangle$ for the eigenstates $|j,\omega\rangle$ of $J_\omega:=\omega_xJ_x+\omega_yJ_y+\omega_zJ_z$, $P_{j,\omega}\cong\eins_{t,j}\otimes|j,\omega\rangle_{\mathrm{irr}}\langle j,\omega|$ accounts for the fact that $J_\omega\psi=j\psi$ defines a subspace of dimension $\mu_{j,t}$ depending only on $j$ and $t$. Equivalently, $\mu_{j,t}$ is the multiplicity of the spin-$j$ irreducible representation within $\mathcal H_t$. Since the maximal eigenvalue $j$ of $J_\omega$ is nondegenerate within each spin-$j$ irreducible, the corresponding eigenspace in $\mathcal H_t$ has dimension $\mu_{j,t}$. Therefore, the coherent-state resolution of the identity \cite[Eqn.~3.16]{arcchiAtomicCoherentStatesInQuantumOptics} yields
        \begin{align}
            \int_{S^2} M_{j,t}(\omega)d\omega=E_tE_j.
        \end{align}
        Given measurement outcomes $(t,j)$ of total photon number and spin, the measurement outcomes $\omega=(\omega_x,\omega_y,\omega_z)$ fulfill on average over the $(j,t)$ outcomes the equality 
        \begin{align}
            \mathbb E\omega_i=n_i\Tr(\tfrac{J_\mathbf{n}}{J+1}\rho),
        \end{align}
        where $\mathbf n=(\sin(2\theta)\cos(\phi),\sin(2\theta)\sin(\phi),\cos(2\theta))$ and $J_\mathbf{n}=\sum_in_iJ_i$ (see Lemma \ref{lem:first-and-second-order-spin-momemnts}), and thus allow estimation of the different collective spin directions. For conversion to the angles $\hat\theta$, $\hat\phi$ we use the estimators 
        \begin{align}\label{def:(theta,phi)-estimator}
            \hat\theta=\arccos(\clip(\tfrac{j+1}{j}\hat \omega_{j,t,z}))/2,\qquad \hat\phi=\mathrm{arctan2}(\hat \omega_{j,t,y},\hat \omega_{j,t,x}),
        \end{align}
        where $\clip(x):=x$ if $x\in[-1,1]$, while $\clip(x)=-1$ once $x<-1$ and $\clip(x)=1$ for $x>1$ or $x=\infty$ (which happens when $j=0$).
        As we show in appendix \ref{app:holevo-optimality}, this measurement is asymptotically optimal in the Holevo sense. 

        After estimating numbers $\hat M, \hat N, \hat\theta, \hat\phi$, the unit cell updates its associated beam-splitter and phase-shifter values by $-\hat\theta+\pi/2$ and $-\hat\phi$
        .

    
    \subsubsection{Neural Network Perspective}
        Let $\mathcal L=(\mathcal L_t)_{t=1,\ldots,}$ with each $\mathcal L_t$ containing a number of indices $((i_1,j_1),\ldots, )$ such that the corresponding $2\times2$ matrices commute.  For example, the list can be given by 
        \[
            \begin{array}{cccccc}
                [ 1 & 2 ] & [ 3 & 4 ] & [ 5 & 6 ] \\
                2  & [ 1 & 4 ] & [ 3 & 6 ] & 5\\
                {}[ 2 & 4 ] & [ 1 & 6 ] & [ 3 & 5 ] \\
                4&[2&6]&[1&5]&3\\
                {}[4&6]&[2&5]&[1&3]\\
                6&[4&5]&[2&3]&1
            \end{array}
        \]
        where $[i,j]$ denotes joint processing of nodes $i$ and $j$ by a unit cell, and the schedule equals the one proposed in \cite{Whiteside1984ParallelJacobi}, which also fits exactly to the layout of today's universal optical chips.

        For each $\mathcal L_t$, the network executes an algorithm which applies the estimation procedure to each mode pair $(i,j)\in\mathcal L_t$, yielding estimates $(\hat M_i,\hat N_j,\hat\theta_{ij},\hat\phi_i)$. These are used to set the parameters $\theta_{\mathrm{cell}},\phi_{\mathrm{cell}}$ to $-\hat\theta_{ij},-\hat\phi_i$. If $\hat\theta=\theta$ and $\hat\phi=\phi$, then the correlation matrix $C=[\Tr(a^\dagger_ia_j\rho)]_{ij}$ satisfies $C_{ij}=0$, and the state after the unit cell has the property $\Tr_{S}(\mathcal C_{ij}\rho\mathcal C_{ij}^\dagger)=S_M\otimes S_N$, where $S=\{l\in[K]:l\neq i\wedge l\neq j\}$ for some parameters $M,N\geq0$. 
        
        This property follows since $\rho$ being a \gls{pun} state is equivalent to it being gauge-invariant \cite[Theorem 3.7]{completePUN}. Hence, the covariance matrix of the reduced state $\Tr_S(\rho)$ is gauge-invariant and thereby the two-mode state $\Tr_S(\rho)$ itself must be a \gls{pun} state. Since every two-mode \gls{pun} state can be written as $US_M\otimes S_NU^\dagger$ for a passive linear unitary $U$, the remaining step is to invoke tools such as the Reck decomposition \cite{reck1994experimental}, which deliver the decomposition $U=V_A(\phi)U_{AB}(\theta)$ since two out of the four parameters of the unitary are irrelevant for the problem due to phase covariance of thermal states and the ambiguity in the choice of a local reference frame for the unit cell. 
        In order to guarantee convergence of the diagonalization method, one can arrange the neural network based on the structure proposed in \cite{Whiteside1984ParallelJacobi}. This earlier proposal describes a realization of the Jacobi algorithm for software.
        We note that, while our later mathematical formulation suggests an algorithmic procedure, the entire proposed procedure can alternatively be considered as a series of ``sweep blocks'' which iteratively perform the same complete sweep several times. In this perspective, each unit cell runs the procedure of estimating the incoming light and then adapting its own parameters exactly once.

\section{Main Results}
    Our first result describes the performance of a universally optimal collective measurement, which consists of three commuting measurements but needs exact knowledge of the parameter $\phi$. To simplify conversion between estimators, we add assumptions on the parameters.
    \begin{theorem}\label{thm:main}
        Let $\theta\in(0,\pi/2)$ as well as $M-N\geq\Delta$ and $N\geq\Delta$ for some $\Delta,\epsilon>0$. The covariance matrix of the unit cell measurement applied to the state $\rho=V_A(\phi)U_{AB}(\theta)(S_M\otimes S_N)U_{AB}(\theta)^\dagger V_A(\phi)^\dagger$ is given by 
        \begin{align}
            \mathrm{cov}=\frac{1}{k}\left(\begin{array}{cccc}
                            M(M+1)&0&0&0\\
                            0&N(N+1)&0&0\\
                            0&0&\tfrac{2M(N+1)}{4(M-N)^2}&0\\
                            0&0&0&\tfrac{2M(N+1)}{\sin(2\theta)^2(M-N)^2}.
            \end{array}\right)+o(1/k)
        \end{align}
        Second, the measurement attains the Holevo CR bound $C^H$ \cite[Eqn. (6)]{ragy2016compatibility} for the cost matrix $G=\mathrm{diag}(1,1,1,\sin^2(2\theta)/4)$, which equals $C^H=M(M+1) + N(N+1) + M(N+1)/(M-N)^2 $. 
    \end{theorem}
    Our second result describes the performance of a measurement consisting again of three commuting collective measurements. However, in this case, the third measurement yields the vector $\omega\in\mathbb R^3$, which can be used to estimate $(\theta,\phi)$ according to \eqref{def:(theta,phi)-estimator}.
    \begin{theorem}\label{thm:restricted-case}
        Let $\theta\in(0,\pi/2)$ as well as $M-N\geq\Delta$ and $N\geq\Delta$ for some $\Delta,\epsilon>0$. The estimators for $M$, $N$ and $\theta$ are consistent. The covariance matrix of the unit cell measurement for the restricted model, where $\rho = U(\theta)S_M\otimes S_N U(\theta)^\dagger$, is given by 
        \begin{align}
            \mathrm{cov}=\frac{1}{k}\left(\begin{array}{ccc}
                            M(M+1)&0&0\\
                            0&N(N+1)&0\\
                            0&0&\tfrac{2MN+M+N}{4(M-N)^2}\\
            \end{array}\right)+o(1/k)
        \end{align}
        Second, the measurement is optimal in the sense that it saturates the Holevo CR bound for all cost matrices $G$. 
    \end{theorem}
    To quantify the accuracy and convergence properties of the proposed neural network algorithm, we use the machinery typically employed to analyze the Jacobi algorithm \cite{932f0815-1804-3f94-9d56-120cca8e1385, hari2017convergencecyclicquasicyclicblock}. A sweep is a (partially parallel) application of unitary $2\times2$ rotations, such that every mode pair $(i,j)$ ($i<j$) is addressed once. Thus, the number of Jacobi rotations $m$ in each sweep is given by
    \begin{equation}
        m:=\frac{K(K-1)}{2}. \label{eq:jac_number_pivots}
    \end{equation}

Let $C^{(\ell)}\in\mathbb{C}^{K\times K}$ be the Hermitian correlation matrix at the
beginning of sweep $\ell$, where $K$ denotes the number of modes. 
 Element wise, $C:=\big(\Tr[a_i^\dagger a_j\rho]\big)_{i,j}$. Since a two-mode covariance is diagonalized by a single rotation, we assume $K\ge3$.
Let $C_{\ell,t}$ denote the matrix after $t$ of the $m=K(K-1)/2$ pivot
operations in that sweep, so that
\begin{equation}
    C_{\ell,0}=C^{(\ell)}, \qquad C_{\ell,m}=C^{(\ell+1)}.
\end{equation}
Let $\off(X):=X-\diag(X)$ denotes the off-diagonal part of a Hermitian matrix $X$, and define off-diagonal residuals $s_\ell$ as:
    \begin{align}
           s_{\ell,t} := \Big( \sum_{i\neq j} |(C_{\ell,t})_{ij}|^2 \Big)^{1/2}, \qquad s_\ell:=s_{\ell,0}. \label{eq:jac_offnorm}
    \end{align}
Let $z_{\ell,t}$ and $r_{\ell,t}$ denote the pivot before and after the
operation at step $t$, respectively. Complex pivots are phase-aligned before the corresponding real Jacobi rotation is applied. Let $\lambda_{1},\ldots,\lambda_{K}$ be the eigenvalues of $C^{(\ell)}$, which in the specific case studied here are given by $N_1,\ldots,N_K$. Then the minimum spectral gap $\Delta$ is $\Delta := \min_{a\neq b}|\lambda_{a}-\lambda_{b}|$, which remains unchanged throughout each sweep and, by extension, for each application of the algorithm. Furthermore, intermediate eigenvalues of the pivot blocks are also interlaced by the invariant spectrum as per the Cauchy-interlacing theorem \cite{fisk2005shortproofcauchysinterlace}. 

\paragraph{Assumptions:}
To guarantee convergence, we need to make the following assumptions:
    \begin{itemize}
      \item[\textup{(A1)}] \emph{Nondegeneracy:} $\Delta>0$.
      \item[\textup{(A2)}] \emph{Local basin:} The initial off-diagonal norm satisfies $s_0 < \frac{\Delta}{K^2}.$
    \end{itemize}

\paragraph{Threshold procedure:}
 If $|z_{\ell,t}|\le\epsilon_{\mathrm{loc}}$, the pivot is skipped and
$r_{\ell,t}=z_{\ell,t}$. Otherwise, an angle
    $\widehat\theta_{\ell,t}\in[-\pi/4,\pi/4]$ is applied such that
    $|r_{\ell,t}|\le\epsilon_{\mathrm{loc}}$. Next, we define a global stoppage criterion, $\epsilon_{\mathrm{glob}}
    :=
    3K\epsilon_{\mathrm{loc}}$, and a stoppage sweep as:
\begin{equation}
     L
    :=
    \inf\left\{
        \ell\ge0:
        s_\ell\le\epsilon_{\mathrm{glob}}
    \right\}.
\end{equation} 

\begin{theorem}[Run-time of the algorithm]
    Assume \textup{(A1)--(A2)}, and the threshold procedure, choose $ 0<\epsilon_{\mathrm{loc}}
    \le
    \frac{\Delta}{6K^3}.$ Then, the stoppage sweep $L$ is
    \begin{equation}
         L
        =
        O\!\left(
            \log\log
            \frac{\Delta}{K^3\epsilon_{\mathrm{loc}}}
        \right)
    \end{equation}
such that, $  s_L
    \le
    \epsilon_{\mathrm{glob}}
    =
    3K\epsilon_{\mathrm{loc}}$. 
\label{thm_convergence_analysis}
\end{theorem}

The next theorem states the number of copies required by the sensing-assisted Jacobi algorithm.

\begin{theorem}[Sample complexity]
\label{thm:sample_complexity}
Let $\epsilon_{\mathrm{loc}}>0$ be the prescribed local residual accuracy, and let $\alpha_{\mathrm{loc}}\in(0,1)$ be the failure probability assigned to each unit-cell, which is related to the number $k$ of resource states consumed in each run of a unit cell via $\exp(-k\cdot C_\star\cdot \epsilon_\mathrm{loc}^2+\mathcal O(\log k))\leq \alpha_\mathrm{loc}$. A sufficient number of copies for the protocol to succeed is: 
   \begin{equation}
     N\!\left(
        \epsilon_{\mathrm{loc}},
        \alpha_{\mathrm{loc}}
    \right)
    = \mathcal O\!\left[ \frac{K}{C_\star\epsilon_{\mathrm{loc}}^{2}}
        \log\log\!\left(
            \frac{1}{\epsilon_{\mathrm{loc}}}
        \right)  \Bigg\{
            \log\!\left(
                \frac{1}{\alpha_{\mathrm{loc}}}
            \right) + D
        \Bigg\}
    \right], \quad D = 
            \log\!\left(
                1+
                \frac{1}
                     {C_{\star}\epsilon_{\mathrm{loc}}^{2}}
                \log\!\left(
                    \frac{1}{\alpha_{\mathrm{loc}}}
                \right)
            \right)
    \label{eq:total_copy_complexity_instance_local}
\end{equation}

\end{theorem}


Different sensing mechanisms of a unit cell are possible, out of which we highlight four. The first one is heterodyne detection, the second one homodyne detection. As the third method we list the estimation procedure based on the \gls{povm} \eqref{def:omega-povm}. 
The fourth method is derived based on the question how to harness the quantum advantage identified during the development of the third method. This question can be separated into three parts: First, which parameter region is the one where quantum technology has provable advantage? Second, how can the optimal measurement be constructed in this domain? Finally, is today's technology sufficient?

As it turns out, the weak-light regime is the one where the proposed algorithm shows greatest potential. In this regime however, detection events are sparse. As we show below, this leads to a situation where unit cells that perform vacuum detection are close to optimal in the region of primary interest.

For implementation using state of the art technology, we assume three basic settings for the unit cell. Each measurement is given by a specific passive unitary transformation followed vacuum detection on both modes. The passive unitary transformations $U_x,U_y,U_z$ are defined by the interferometer settings $U_x=U(\pi/4)V_A(0)$, $U_y=U(\pi/4)V_A(-\pi/2)$, $U_z=U(0)V_A(0)$ while the vacuum detectors are modeled by $\{|0\rangle\langle0|,\eins-|0\rangle\langle0|\}$. 

Upon input of the state $\rho(M,N,\nu,\nu')$ the corresponding probabilities of the four different detection events can be calculated using \cite[Eqn. (7)]{Kim:20}:
\begin{align}
    p(0,0) &= \frac{1}{(M+1)(N+1)},\\
    p(0,1) &= \frac{\cos^2(\nu)N+\sin^2(\nu)M+MN}{(M+1)(N+1)},\\
    p(1,0) &= \frac{\cos^2(\nu)M+\sin^2(\nu)N+MN}{(M+1)(N+1)}. 
\end{align}
For the three different settings defined above, the respective values of $\cos^2(\nu_i)$ are given by $(1+\sin(2\theta)\cos(\phi))/2$, $(1+\sin(2\theta)\sin\phi))/2$ and $\cos^2(\theta)$. Thus the estimator 
\begin{align}
    \tilde{\mathbf n} = (\hat n_{A,x}-\hat n_{B,x}, \hat n_{A,y}-\hat n_{B,y}, \hat n_{A,z}-\hat n_{B,z})
\end{align}
which utilizes a total of $k=3$ copies of $\rho(M,N,\theta,\phi)$ to produce one estimate delivers, by averaging and normalization, an estimate $\omega$ of $\mathbf n=(\sin(2\theta)\cos(\phi),\sin(2\theta)\sin(\phi),\cos(2\theta))$. The proposed method now utilizes \gls{tdma} to deal with the incompatibility of measurements in the different spin directions. The convergence to the true value of $\mathbf n$ can be derived based on the Chernoff bound and is the content of Theorem \ref{thm:unit_cell_performance}. 

\begin{theorem}[Unit Cell Performance] \label{thm:unit_cell_performance}
    Let $\rho=V_A(\phi)U(\theta)(S_M\otimes S_N)U(\theta)^\dagger V_A(\phi)^\dagger$ with unknown parameters $N<M$, $\phi\in[0,2\pi)$, $\theta\in[0,\pi/2)$. Assume there are known constants $\Delta>0$, $B>0$ such that $M-N\geq\Delta>0$, $N>\Delta$ and $M<B$. 
    Assume the vector $\mathbf n = (\sin(2\theta)\cos(\phi),\sin(2\theta)\sin(\phi),\cos(2\theta))$ is estimated with a target precision of $1-\sqrt{7/8}>\epsilon>0$. Each method has a different asymptotic scaling denoted by $C_\mathrm{het}$ for heterodyne, $C_{\mathrm{hom}}$ for homdoyning, $C_\mathrm{spovm}$ for the spin \gls{povm} and $C_\mathrm{tdmapc}$ for the \gls{tdma} photon counting method. For the spin \gls{povm} we have a binary decision $\{\mathrm{abort},\mathrm{continue}\}$ where 
    \begin{align}
        \mathbb P(\|\omega-\mathbf n\|_2\geq\epsilon|\mathrm{continue})&\leq\exp(-k\cdot C\cdot\epsilon^2+\mathcal O(\log k))\label{eqn:standard-formula-for-error-exponent}\\
        \mathbb P(\mathrm{abort})&\leq\exp(-k\cdot C\cdot\epsilon^2+\mathcal O(\log k)).
    \end{align}
    For the other methods it holds $\mathbb P(\mathrm{abort})=0$. The constant $C$ is an element of $\{C_\mathrm{het},C_{\mathrm{hom}},C_\mathrm{spovm},C_\mathrm{tdmapc}\}$ and the latter constants are given by
    \begin{description}[
        align=left,
        labelwidth=3cm,
        leftmargin=3.5cm,
        labelsep=0.5cm
    ]
        \item[Heterodyne] $C_\mathrm{het}=\tfrac{\Delta^2}{3(B+1)^2}$
        \item[Homodyne] $C_\mathrm{hom}=\tfrac{\Delta^2}{48(B+1/2)^2}$
        \item[Spin \gls{povm}] $C_\mathrm{spovm}=\tfrac{\Delta^2}{16B(B+1)}$
        \item[\gls{tdma} Photon Counting] $C_\mathrm{tdmapc}=\tfrac{\Delta^2}{144B(B+1)^3}$
    \end{description}
    Further, an estimate $\omega$ of $\mathbf n$ that satisfies $\|\omega-\mathbf n\|\leq\epsilon$ leads to a local error of $|\Tr(a_i^\dagger b\rho')|\leq B\epsilon/2$, where $\rho'$ is the state after the unit cell applied the unitary rotation corresponding to $\omega$ on $\rho$.

\end{theorem}

\begin{figure}[t]
    \centering
    \includegraphics[width=\linewidth]{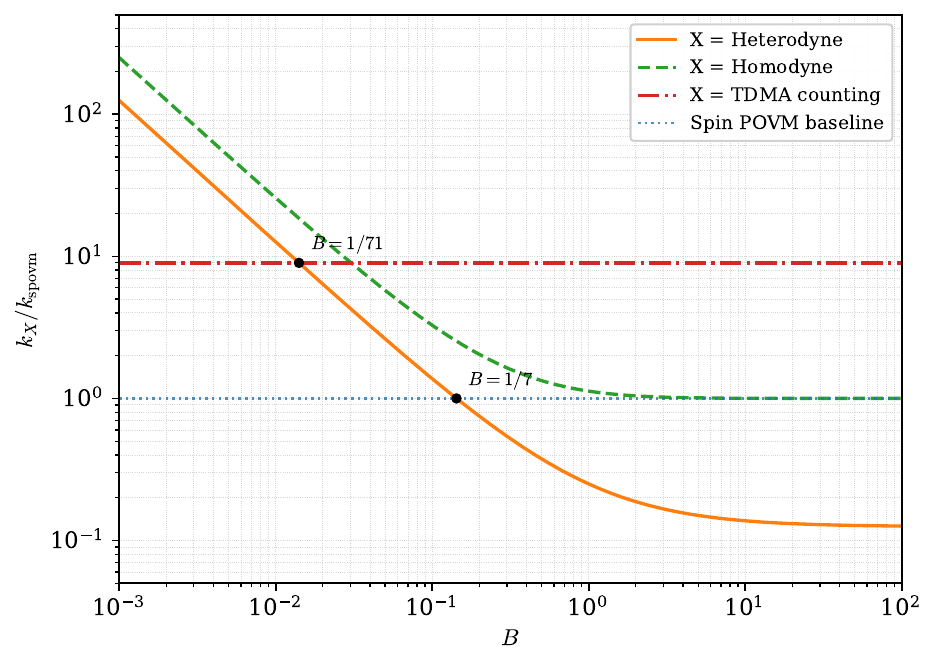}
    \caption{Unit-cell sample-complexity comparison. The ratio $k_X/k_{\rm spovm}$
    of copies required by method $X$ and by the spin POVM is shown against $B$. Other constants such as the gap $\Delta$, precision $\varepsilon$ and confidence
    $\alpha$ cancels in the ratio, so the curves are universal. Similar regions of advantage were already observed in \cite{noetzelMunarVallespir}.}
    \label{fig:unitCell_comparison}
\end{figure}


The first immediate observation is that, for small bounds $B\ll1$, the shot noise contributions in the denominators of $C_\mathrm{het}$ and $C_\mathrm{hom}$ prevent the respective exponent $C$ from scaling beyond a constant, while the exponents of the spin \gls{povm} and the \gls{tdma} photon counting grow as $\sim1/B$, thus leading to increased performance. 

The second observation we make is that, while the photon counting estimator is worse than the optimal one by almost an order of magnitude, it captures the essential scaling extremely well, with a denominator of $144B(B+1)^3$ versus the $16B(B+1)$ we derived for the fully coherent spin \gls{povm} \eqref{def:(theta,phi)-estimator}. For the small values of $B$ we are interested in, the added effort for realizing measurements coherently over many blocks appears debatable given these scalings.

\paragraph{Total-network comparison: }
To properly compare our approach to the state of the art and obtain a clear view of the actual region of advantage, we need one more comparison: the one with a system that performs heterodyne detection on all $K$ received modes, followed by analysis of the covariance matrix using conventional computing systems. Under the favorable assumption of perfect operation for such a device, its only error is due to statistical deviations. The statistical properties of Gaussian distributions are well known and therefore underpin the following analysis.

A cyclic sweep visits every pivot and each rotation refills the rows and columns it touches, so a total of $K(K-1)/2$ off-diagonal
entries. Aggregating a uniform per-entry bound in the sense of the off-diagonal error \eqref{eq:jac_offnorm} gives $\epsilon_{\mathrm{glob}}\leq\sqrt{K(K-1)}\,\epsilon_{\mathrm{loc}}$, so
\begin{equation}\label{eq:local-tolerance}
    \epsilon_{\mathrm{loc}} = \frac{\epsilon_{\mathrm{glob}}}{\sqrt{K(K-1)}}
\end{equation}
suffices to certify a desired target accuracy. Suppressing logarithmic factors, $\widetilde{O}(\cdot)$, each unit-cell based approach has, according to \eqref{eq:total_copy_complexity_instance_local}, a resource consumption of 
\begin{align}
    N_\mathrm{sens}(\epsilon_\mathrm{glob},\alpha) = \widetilde{O}\!\left(\frac{K^{3}B^{2}}
        {C_\mathrm{sens} \cdot \epsilon_{\mathrm{glob}}^{2}}
    \right),
\end{align}
where $\mathrm{sens}\in\{\mathrm{het},\mathrm{hom},\mathrm{spovm},\mathrm{tdmapc}\}$ is one of the four methods discussed in Theorem \ref{thm:unit_cell_performance} and we replaced the (unknown) spectral width $W$ with the (known) upper bound $B$. In comparison, the expected error using heterodyne detection directly on all $K$ modes satisfies 
\begin{equation}
    N_{\mathrm{g-het}} = \widetilde{O}\!\left(\frac{K^{2}(B+1)^{2}}{\epsilon_{\mathrm{glob}}^{2}}\right),
\end{equation}
 For the proof, see Eq. 14 of \cite{fanizza2026efficienthamiltonianstructuretrace} or Appendix E of \cite{Bittel_2025}. For $\mathrm{sens}\in\{\mathrm{hom},\mathrm{het}\}$ the ratio of the two costs is then given by
\begin{align}
    \frac{N_{\mathrm{het}}}{N_{\mathrm{g-het}}} =\widetilde{O}\!\left( K\Big(\frac{B}{\Delta}\Big)^2 \right)  \qquad \frac{N_{\mathrm{hom}}}{N_{\mathrm{g-het}}}  = \widetilde{O}\!\left( K\frac{B^2(B+1/2)^2}{(B+1)^2\Delta^2} \right)
\end{align}
For the case of $\mathrm{sens}\in\{\mathrm{spovm},\mathrm{tdmapc}\}$ we have 
\begin{equation}\label{eq:ratio}
    \frac{N_{\mathrm{spovm}}}{N_{\mathrm{g-het}}} = \widetilde{O}\!\left[  K\,\frac{B^{3}}{\Delta^{2}}(B+1)^{-1}
    \right],\qquad \frac{N_{\mathrm{tdmapc}}}{N_{\mathrm{g-het}}} = \widetilde{O}\!\left[  K\,\frac{B^{3}}{\Delta^{2}}(B+1)
    \right].
\end{equation}
Since $\Delta/B<1$ must hold, this shows that both the $\mathrm{spovm}$ and the $\mathrm{tdmapc}$ unit cell-based approach provides an advantage which scales as $\sim1/B$ when $B$ goes to zero. However, the price for this advantage is coming at an increase in cost which scales linear with the number $K$ of modes. As pointed out before, the main advantage of the $\mathrm{spovm}$- based method over the $\mathrm{tdmapc}$ one is a factor of $9$ reduction in the error exponent \eqref{eqn:standard-formula-for-error-exponent}, which is however invisible in above analysis.

\begin{figure}[t]
    \centering
    \includegraphics[width=\linewidth]{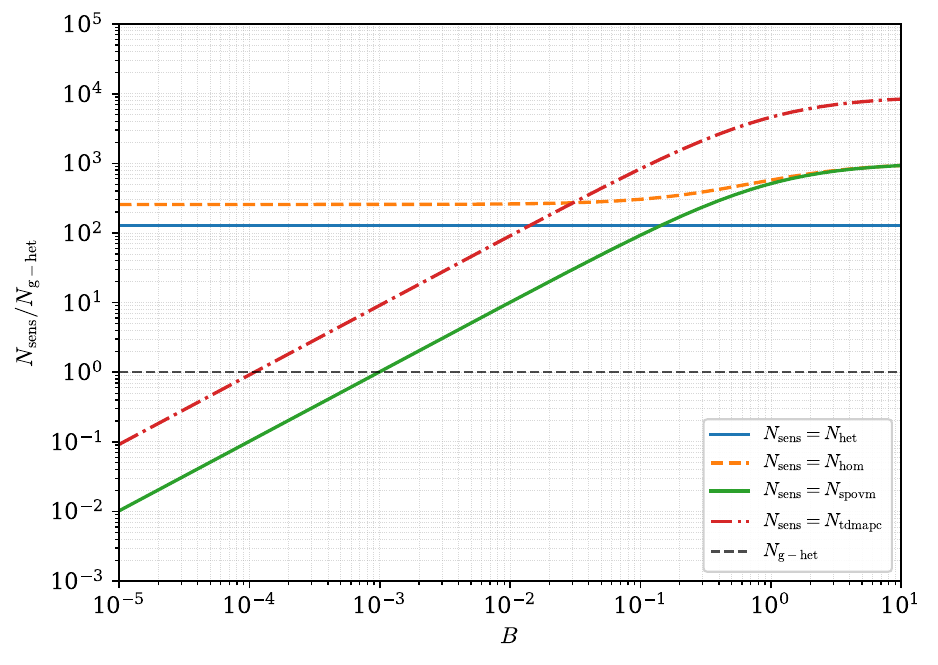}
    \caption{ The ratio
    $N_{\rm sens}/N_{\rm g-het}$ of the total number of copies required by the
    unit-cell-based Jacobi algorithm and by global heterodyne tomography on all
    $K$ modes is shown as a function of the brightness $B$, for $K=16$. The dashed line marks the break-even point
    $N_{\rm sens}=N_{\rm g-het}$, and curves below it indicate a region of potential quantum
    advantage. $N_{\rm tdmapc}$  and $N_{\rm spovm}$ reaches this region at low brightness.}
    \label{fig:network_comparison}
\end{figure}

Finally, in the weak-photon limit, \(B\ll 1\), we obtain
\begin{equation}
    N_{\mathrm{g\text{-}het}}
    =
    \widetilde{O}\!\left(
        \frac{K^{2}}
        {\epsilon_{\mathrm{glob}}^{2}}
    \right),
    \qquad
    N_{\mathrm{spovm}}
    =
    \widetilde{O}\!\left(
        \frac{K^{3}B^{3}}
        {\Delta^{2}\epsilon_{\mathrm{glob}}^{2}}
    \right).
\end{equation}
Indeed, using \(B/(B+1)\simeq B\), the ratio
\eqref{eq:ratio} reduces to $ B
    \lesssim
    \left(\frac{\Delta^{2}}{K}\right)^{1/3}.$ The heterodyne cost has saturated at its vacuum-noise floor, while the $\mathrm{spovm}$ and $\mathrm{tdmapc}$ cost still fall linearly in $B$, which is what opens the opportunity for quantum technology to provide advantage.

These different ratios are plotted in Fig. \ref{fig:network_comparison}. The dashed line marks the break-even point $N_{\rm sens}=N_{\rm g-het}$; curves below it indicate a window of quantum advantage. Only the spin POVM and TDMA photon counting reach it, and only at the weak photon limit.

\section{Proofs of the Main Statements}
  We start with the proof of Theorem \ref{thm:restricted-case}, during which we utilize the main technical ingredients, which we then extend to cover the proof of Theorem \ref{thm:main}. Afterward, in section \ref{sec:convergence_and_complexity_analysis}, we prove the convergence guarantees, which are the content of Theorem \ref{thm_convergence_analysis}. Finally, we prove the sample complexity analysis Theorem \ref{thm:unit_cell_performance} in section \ref{sec:quantum_advantage}. The proof of these theorems builds on several lemmas. For readability, the proof of these lemmas is shifted to Appendix \ref{appendix:technical_proofs}.
\subsection{Proof of Theorem \ref{thm:main} and Theorem \ref{thm:restricted-case} }

    \begin{proof}[Proof of Theorem \ref{thm:restricted-case}]
        We structure the proof into the three statements of the theorem, starting with upper bounds on the entries of the covariance matrix. As we show later, these bounds match at first order with the lower bounds provided by \cite{ragy2016compatibility}, thus deriving the upper bounds will be sufficient to proving the claimed structure of $\cov$.
            \paragraph{Covariance Matrix:} For the covariance bounds of Theorem \ref{thm:restricted-case}, we use \glspl{pvm} and therefore write the covariance of obtaining outcomes $(A,B))$, which is linear in both its entries, as
                \begin{align}
                    \cov(X,Y)=\Tr(\rho XY) - \Tr(\rho X)\Tr(\rho Y).
                \end{align}
                For the off-diagonal terms, we observe that by the Cauchy-Schwarz inequality we have
                \begin{align}
                    \cov(X,X)&\leq\Tr(\rho X^2)\label{ineq:var}\\
                    |\cov(X,Y)|&\leq \sqrt{\Tr(\rho X^2)\Tr(\rho Y^2)}\label{ineq:cov}.
                \end{align} 
                This allows us to prove (upon defining $D_z:=J-J_z$) bounds such as: 
                \begin{align}
                    \cov&(\hat M,\hat N)=\tfrac{1}{k^2}\cov(\tfrac{1}{2}N_{AB}+J,\tfrac{1}{2}N_{AB}-J)\\
                        &=\tfrac{1}{k^2}\cov(N_{A}+D_z,N_{B}-D_z)\\
                        &= \tfrac{1}{k^2}\left[\cov(N_{A},N_{B})-\cov(N_{A},D_z)+\cov(D_z,N_{B})-\cov(D_z,D_z)\right]\\
                        &\leq \tfrac{1}{k^2}\left[|\cov(N_{A},D_z)|+|\cov(D_z,N_{B})|+|\cov(D_z,D_z)|\right]\\
                        &\leq \tfrac{1}{k^2}\left[\sqrt{\svar(N_{A})\svar(D_z)}+\sqrt{\svar(D_z)\svar(N_{B})}+\svar(D_z)\right]\\
                        &\leq \tfrac{1}{k^2}\left[\sqrt{k(M^2+M)}+\sqrt{k(N^2+N)}+\sqrt{\Tr(D_z^2\rho_0)}\right]\sqrt{\Tr(D_z^2\rho_0)}\\
                        &=\mathcal O(k^{-3/2})
                \end{align}
                where the last line follows since both $N_{AB}$ and $J$ commute with $U(\theta)^{\otimes k}$ such that everything can be computed relative to $\rho_0=\rho(M,N,\theta)$, and in addition from the following Lemma:
                \begin{lemma}\label{lem:bound-on-Dz}
                    For the operator $D_z=J - J_z$ it holds for all $t\geq1$
                    \begin{align}
                        \Tr(D_z^tS_N\otimes S_M) &= \mathcal O(1).
                    \end{align}
                    in particular the constant is given by the polylogarithm $L_{-t}(1/q)$ \cite{polylogarithm} for $q=M(N+1)/N(M+1)$. Further, we have 
                    \begin{align}
                        \partial_M\Tr(D_z\rho_0)&=\mathcal O(\sqrt{k}),\qquad \partial_N\Tr(D_z\rho_0)=\mathcal O(\sqrt{k}).
                    \end{align}
                \end{lemma}
                We note that Lemma \ref{lem:bound-on-Dz} already establishes consistency and local unbiasedeness of the estimators $\hat M$ and $\hat N$ \cite{DemkowiczGoreckiGuta2020} by the commutation rules \ref{eqn:commutators}. 
                The variance of $\hat M$ and $\hat N$ can be bounded by noting again that the measurement commutes with $U(\theta)^{\otimes k}$, so that all calculations can be carried out with respect to $\rho_0$. 
                \begin{align}
                    \mathrm{var}(\hat M)
                        &=\svar(\hat M)_{\rho_0}\\
                        &=\tfrac{1}{k^2}\mathrm{var}(N_A+D_z)\\
                        &=\tfrac{1}{k^2}[\mathrm{var}(N_A)+2\cdot\mathrm{cov}(N_A,D_z)+\svar(D_z)]\\
                        &\leq\tfrac{1}{k^2}[\mathrm{var}(N_A)+2\sqrt{\mathrm{var}(N_A)\mathrm{var}(D_z)}+\Tr(D_z^2\rho_0)]\\
                        &\leq\tfrac{1}{k^2}[k\cdot M(M+1) + \mathcal O(\sqrt{k})]
               \end{align}
               where we used the fact that $\svar(S_M)=M(M+1)$ as well as Lemma \ref{lem:bound-on-Dz} and the elementary properties of $\cov$ \eqref{ineq:cov} and $\svar$ \eqref{ineq:var}. By applying similar steps, we arrive at 
               \begin{align}
                    \svar(\hat M) &= \tfrac{1}{k}\cdot M(M+1) + o(1/k)\\
                   \mathrm{var}(\hat N)
                        &=\tfrac{1}{k}\cdot N(N+1) + o(1/k).
               \end{align}
               For the variance of the estimator of $\theta$ alone we start first with that for $\tau$. We employ the following Lemma \ref{lem:alternative-representations}
                \begin{lemma}[Alternative Representations]\label{lem:alternative-representations}
                    The measurement operators can be expressed in the following alternative forms:
                    \begin{align}
                        \hat M &= (N_{A}+D_z)/k\label{eqn:hatM}\\
                        \hat\tau &= \tau + U(\theta)^{\dagger\otimes k}\big(\tfrac{s}{2}J_xJ^{-1}-\tfrac{c}{2}D_zJ^{-1}\big)U(\theta)^{\otimes k}\label{eqn:hatTau}
                    \end{align}
                    where $J_x:=(J_++J_-)/2$ and $D_z:=J - J_z$.
                \end{lemma}
                which again allows us to calculate the variance with respect to $\rho_0=\rho(M,N,0)$ instead of $\rho(M,N,\theta)$:
               \begin{align}
                   \svar(\hat\tau)&=\svar(\tau + U(\theta)^{\dagger\otimes k}\big(\tfrac{s}{2}J_xJ^{-1}-\tfrac{c}{2}D_zJ^{-1}\big)U(\theta)^{\otimes k})\\
                        &=\svar(\tau + \big(\tfrac{s}{2}J_xJ^{-1}-\tfrac{c}{2}D_zJ^{-1}\big))_{\rho_0}\\
                        &=\Tr(\big(\tfrac{s}{2}J_xJ^{-1}-\tfrac{c}{2}D_zJ^{-1}\big)^2\rho_0)\\
                        &=\frac{s^2}{4}\Tr(J_x^2J^{-2}\rho_0) + \frac{c^2}{4}\Tr(D_z^2J^{-2}\rho_0) - \frac{cs}{4}\Tr(J_xJ^{-1}D_zJ^{-1}\rho_0)\\
                        &\overset{(a)}{\leq} \frac{s^2}{4}\frac{2MN+M+N}{M-N}\frac{1}{2}\Tr(J^{-1}\rho_0) + o(1/k) - 0\\
                        &\overset{(b)}{\leq} \frac{s^2}{4}\frac{2MN+M+N}{k(M-N-\epsilon)^2} + o(1/k)\label{eqn:var(tau)-estimate}
               \end{align}
               where $(a)$ follows from Lemma \ref{lem:bound-onJx2J-2} and $(b)$ from Lemma \ref{lem:bound-for-Jinverse}. Since the statement holds for all $\epsilon>0$, the required upper bound is proven.
                \begin{lemma}\label{lem:bound-for-Jinverse}
                    For every $\epsilon>0$ satisfying in addition also $M-N-\epsilon>0$ there is a $K_0\geq0$ and a $C_{M,N,\epsilon}>0$ such that for the Hermitian operator $J^{-1}$ defined by $J^{-1}J|\psi\rangle=|\psi\rangle$ whenever $J|\psi\rangle\neq0$ and $J^{-1}|\psi\rangle=0$ whenever $J|\psi\rangle=0$ the inequality
                    \begin{align}
                        \Tr(J^{-t}\rho_0) &\leq \frac{2^t}{k^t(M-N-\epsilon)^t} + e^{-k\cdot C'_{M,N,\epsilon}}\label{eqn:JinverseBound}
                    \end{align}
                    holds for $t=1,\ldots,6$ and all $k\geq K_0$.
                \end{lemma}
                \begin{lemma}\label{lem:bound-onJx2J-2} Let $D_z:=J-J_z$ and $J_x:=\frac{1}{2}(J_++J_-)$. Then for $\rho_0=S_M{^\otimes k}\otimes S_N^{\otimes k}$ and $t=1,2,3,4$ we have
                    \begin{align}
                        \Tr(J_x^2J^{-2}\rho_0)&\leq\frac{2MN+M+N}{M-N}\frac{1}{2}\Tr(J^{-1}\rho_0) \label{eqn:bound-onJx2J-2}\\
                        \Tr(J_x^4J^{-4}\rho_0)&\leq\frac{3}{4}\left(\frac{2MN+M+N}{M-N}\right)^2\Tr(J^{-2}\rho_0) \label{eqn:bound-onJx4J-4}\\
                        \Tr(D_z^tJ^{-t}\rho_0)&\leq \mathcal O(k^{-t})\label{eqn:bound-onD_z2J-2}
                    \end{align}
                \end{lemma}
                As a final step of our proof of Theorem \ref{thm:restricted-case} we derive a bound on $\cov(\hat M,\hat\tau)$. All other entries of the covariance matrix follow from similar calculations and are not stated.
                Since $\hat M = \tfrac{1}{2}N_{AB} + J$ commutes with $U(\theta)^{\otimes k}$, we have 
                \begin{align}
                    \cov(\hat M, &\hat\tau) = \tfrac{1}{k}\cov(N_A+D_z,\tfrac{s}{2}J_xJ^{-1}-\tfrac{c}{2}D_zJ^{-1})\\
                        &= \tfrac{s}{2k}\cov(N_A,J_xJ^{-1}) - \tfrac{c}{2k}\cov(N_A,D_zJ^{-1}) +\nonumber\\
                        &\qquad + \tfrac{s}{2k}\cov(D_z,J_xJ^{-1}) -\tfrac{c}{2k}\cov(D_z,D_zJ^{-1})\\
                        &= - \tfrac{c}{2k}\cov(N_A,D_zJ^{-1}) + \tfrac{s}{2k}\cov(D_z,J_xJ^{-1}) -\tfrac{c}{2k}\cov(D_z,D_zJ^{-1})
                \end{align}
                since $\langle j,j_z,J_x j,j_z\rangle=0$ for all $(j,j_z)$ and $N_A$, $J$ and $\rho_0$ are diagonal in the same basis. Further by the Cauchy-Schwartz inequality, 
                \begin{align}
                    \tfrac{s}{2k}\cov(D_z,J_xJ^{-1}) &\leq \tfrac{s}{2k}\sqrt{\Tr(D_z^2\rho_0)\Tr(J_x^2J^{-2}\rho_0)}\\
                        &=\mathcal O(k^{-3/2})\\
                        &= o(1/k)
                \end{align}
                according to Lemma \ref{lem:bound-on-Dz} and Lemma \ref{lem:bound-onJx2J-2}. Setting $t=4$ in Lemma \ref{lem:bound-on-Dz} and using Lemma \ref{lem:bound-for-Jinverse} allows us to state 
                \begin{align}
                    \tfrac{c}{2k}|\cov(N_A,D_zJ^{-1})|&\leq\tfrac{c}{2k}\sqrt{\svar(N_A)\svar(D_zJ^{-1})}\\
                        &\leq\tfrac{c}{2k}\sqrt{kM(M+1)}\sqrt{\Tr(D_z^2J^{-2}\rho_0)}\\
                        &\leq\tfrac{c}{2k}\sqrt{kM(M+1)}(\Tr(D_z^4\rho_0)\Tr(J^{-4}\rho_0))^{1/4}\\
                        &\leq\tfrac{c}{2k}\sqrt{kM(M+1)}o(1)(k^{-4})^{1/4}\\
                        &=\mathcal O(k^{-3/2})\\
                        &=o(1/k).
                \end{align}
                Finally for the term $\tfrac{c}{2k}\cov(D_z,D_zJ^{-1})$ we have
                \begin{align}
                    |\cov(D_z,D_zJ^{-1})|&\leq\sqrt{\svar(D_z)\svar(D_zJ^{-1})}\\
                        &\leq \sqrt{\Tr(D_z^2\rho_0)\Tr(D_z^2J^{-2}\rho_0)}\\
                        &\leq \mathcal O(1/k).
                \end{align}
                by \eqref{eqn:bound-onD_z2J-2} of Lemma \ref{lem:bound-onJx2J-2}. Conversion from $\hat\tau$ to $\hat\theta=\arccos(2\hat\tau-1)/2$ is based on the variance bound
                \begin{align}
                    \Tr((\hat\tau-\tau)^4\rho_0) &\leq \Tr(\big(\tfrac{s}{2}J_xJ^{-1}-\tfrac{c}{2}D_zJ^{-1}\big)^4\rho_0)\\
                        &=\Tr(\big(A-B\big)^4\rho_0)\\
                        &=\Tr((A^4 + A^2B^2 + B^4)\rho_0)+\nonumber\\
                        &\ \ + \Tr(ABAB+AB^2A+BA^2B+BABA+B^2A^2)\rho_0)
                \end{align}
                where $A=\tfrac{s}{2}J_xJ^{-1}$ and $B=\tfrac{c}{2}D_zJ^{-1}$ and terms with an odd number of $A$ terms vanished due to the action of $J_x$ and diagonality of $B$ in the photon number basis. We proceed by utilizing Lemma \ref{lem:bound-onJx2J-2} which implies $\Tr(B^t\rho_0)=\mathcal O(k^{-t})$ and $\Tr(A^2\rho_0)=\mathcal O(1/k)$ and $\Tr(A^4\rho_0)=\mathcal O(1/k^2)$ as well as the Cauchy-Schwarz inequality:
                \begin{align}
                    \Tr((\hat\tau-\tau)^4\rho_0) &\leq \mathcal O(k^{-2}) + 2\sqrt{\Tr(A^4\rho_0)\Tr(B^4\rho_0)} + \mathcal O(k^{-4}) +  \nonumber\\&\qquad +\Tr((ABAB+AB^2A+BA^2B+BABA)\rho_0).
                \end{align}
                It remains to bound the terms $ABAB$ and $ABBA$, $BAAB$, $BABA$. For this we use 
                \begin{lemma}[ABBA Lemma]\label{lem:abba}
                    Let $A=\tfrac{s}{2}J_xJ^{-1}$ and $B=\tfrac{c}{2}D_zJ^{-1}$.\\
                    Then for $X\in\{ABAB, BABA, ABBA,BAAB\}$ we have 
                    \begin{align}
                        \Tr(X\rho_0)=\mathcal O(k^{-3}).
                    \end{align}
                \end{lemma}
                Thus, we finally arrive at the conclusion $\Tr((\hat\tau-\tau)^4\rho_0)=\mathcal O(k^{-2})$. We now set $g(\tau):=\tfrac{1}{2}\arccos(2\tau-1)$. We know that $\svar(\tau_k)=\frac{\sin^2(2\theta)}{4}\frac{2MN+M+N}{k(M-N-\epsilon)^2} + o(1/k)$ from \eqref{eqn:var(tau)-estimate} and $\mathbb E|\hat\tau_k-\tau|^4=\mathcal O(1/k^2)$. By Taylor expansion, 
                \begin{align}
                    g(\hat\tau_k)-g(\tau) = -\frac{1}{2\sqrt{\tau(1-\tau)}}(\hat\tau_k-\tau) + R_k
                \end{align}
                where the remainder term is bounded by using $2\tfrac{d^2}{dx^2}\arccos(2x-1)=(1-2x)(x-x^2)^{-3/2}$. We clip the estimator so that $\hat\tau_k\in[k^{-1/4},1-k^{-1/4}]$, leading to
                \begin{align}
                    |R_k| &\leq \frac{\|g''\|_{\infty}}{2}|\hat\tau_k-\tau|^2\\
                    &\leq \frac{2k^{3/8}}{2}|\hat\tau_k-\tau|^2
                \end{align}
                for large enough $k$. It follows, whenever $0<\tau<1$, 
                \begin{align}
                    \mathbb E(R_k^2) =\mathcal O(k^{3/4}\mathbb E(|\hat\tau_k-\tau|^4)) = \mathcal O(k^{-5/4}).
                \end{align}
                Thus we obtain 
                \begin{align}
                    \svar(\hat\theta_k-\theta) &= \tfrac{1}{4\tau(1-\tau)}\svar(\hat\tau_k-\tau) + \svar(R_k)-\tfrac{1}{\sqrt{\tau(1-\tau)}}\cov(R_k,\hat\tau_k-\tau)\\
                        &=\tfrac{1}{4\tau(1-\tau)}\svar(\hat\tau_k-\tau)+o(1/k)+o(1/k),
                \end{align}
                where we used $\cov(\hat\tau_k-\tau,R_k)\leq\sqrt{\svar(\hat\tau_k-\tau)\svar(R_k)}$. Consistency and consistency and local unbiasedness then follow from the preceding- as well as the covariance bounds on $\cov(N_A,D_zJ^{-1})$ (and $\cov(N_A,D_zJ^{-1})$, respectively).
            \paragraph{Holevo bound for the restricted model ($\phi=0$):}
                We utilize the framework derived in \cite{ragy2016compatibility}. It holds $\partial_MS_M=\tfrac{1}{M(M+1)}(N_A-M\eins)S_M$ and $\partial_NS_N=\tfrac{1}{N(N+1)}(N_B-N\eins)S_N$, as well as $\partial_MU(\theta)=\partial_NU(\theta)=0$. Thus, since $[\rho,\partial_M\rho]=[\rho,\partial_N\rho]=0$ the operators $L_M:=\rho^{-1}\partial_M\rho$ and $L_N:=\rho^{-1}\partial_N\rho$ fulfill the equations
                \begin{align}
                    \tfrac{1}{2}(L_M\rho + \rho L_M) = \partial_M\rho,\qquad \tfrac{1}{2}(L_N\rho + \rho L_N) = \partial_N\rho,
                \end{align}
                thus qualifying them for the role of \gls{sld}. We further have $\partial_\theta\rho=[a^\dagger b-ab^\dagger,\rho]$, from where we can conclude that $L_\theta := \tfrac{2(M-N)}{2MN+M+N}(a^\dagger b + ab^\dagger)$ is the \gls{sld} for the parameter $\theta$. The multi-parameter bound on the \gls{sld} Fisher matrix thus reads as \cite[Eq. (4)]{ragy2016compatibility}
                \begin{align}\label{eqn:holevo-bound-1}
                    (\tfrac{1}{2}\Tr(\{L_i,L_j\}\rho))_{i,j\in\{M,N,\theta\}}.
                \end{align}
                A direct calculation shows that it exactly equals the inverse of the $\mathrm{cov}$ matrix of Theorem \ref{thm:restricted-case}. Note that \eqref{eqn:holevo-bound-1} alone does not imply achievability. Rather, this follows from the fact that the \glspl{sld} commute with respect to $\rho$: 
                \begin{align}
                    \Tr([L_i,L_j]\rho))=0,\qquad i\neq j\in\{M,N,\theta\}.
                \end{align}
        \end{proof}
        \begin{proof}[Proof of Theorem \ref{thm:main}]
            Since the measurements for the estimates $\hat M$ and $\hat N$ stay the same, it remains to derive the covariance entries for $\hat\theta$ and $\hat\phi$ for the joint measurement. For this, we use 
            \begin{lemma}\label{lem:first-and-second-order-spin-momemnts}
                Let $E_j$ be the spectral projection to eigenvalue $j$ of $J$ and $E_t$ the one for eigenvalue $t$ of $N_{AB}$. Let $p(j,t):=\Tr(E_jE_t\rho)$. Then 
                \begin{align}
                    \mathbb E_{j,t}\omega_i&=p(j,t)^{-1}\frac{1}{j+1}\Tr J_iE_jE_t\rho\label{eqn:first-order-spin-moment}\\
                    \int\omega_i\omega_lM_{j,t}(d\omega)&= E_jE_t\left[\frac{\{J_i,J_l\}}{(j+1)(2j+3)}+ \delta_{i,l}\frac{1}{2j+3}\eins\right]\label{eqn:second-order-spin-moment}
                \end{align}
            \end{lemma}
            Let $(t,j)$ be given. The conditional variance of the $\omega$ measurement is then, per component $i=x,y,z$, calculated based on the moment operators
        \begin{align}
            B_i^{(j,t)} &= E_t\int\omega_i M_{j,t}(d\omega)E_t = \frac{1}{j+1}E_tE_jJ_i\\
            B_{ii}^{(j,t)} &= E_t\int\omega_i^2 M_{j,t}(d\omega)E_t = \frac{2E_tE_jJ_i^2}{(j+1)(2j+3)}+ \frac{1}{2j+3}E_tE_j
        \end{align}
        which we calculate using Lemma \ref{lem:first-and-second-order-spin-momemnts}. Using these, the variance takes the form
        \begin{align}
            \svar(\omega_i|j,t) &= \frac{1}{p(j,t)}\Tr(B_{ii}^{(j,t)}\rho) -\Tr(B_i^{(j,t)}\frac{1}{p(j,t)}\rho)^2\\
                &= \frac{1}{p(j,t)}\Big[\Tr(\frac{2E_tE_jJ_i^2}{(j+1)(2j+3)}\rho) + \Tr(\frac{1}{2j+3}E_tE_j\rho)-\nonumber\\
                &\qquad\qquad\qquad\qquad\qquad\qquad -\frac{1}{p(j,t)}\frac{1}{(j+1)^2}\Tr(E_tE_jJ_i\rho)^2\Big]\\
                &= \frac{1}{p(j,t)}\Big[\Tr(\frac{2E_tE_jJ_i^2}{(j+1)(2j+3)}\rho) + \Tr(\frac{1}{2j+3}E_tE_j\rho)-\nonumber\\
                &\qquad\qquad\qquad\qquad\qquad\qquad -\frac{1}{p(j,t)}\frac{1}{(j+1)^2}\Tr(E_tE_jJ_i\rho)^2\Big].
        \end{align}
        Thus by the law of conditional variance, we get
        \begin{align}
            \svar&(\omega_i) = \sum_{t,j}\left[\Tr(\Big(\tfrac{2E_tE_jJ_i^2}{(j+1)(2j+3)} + \tfrac{E_tE_j}{2j+3}\Big)\rho)\right]-\Big[\sum_{t,j}\frac{1}{j+1}\Tr(E_tE_jJ_i\rho)\Big]^2\\
            &= \Tr(\big(2J_i^2(J+1)^{-1}(2J+3)^{-1} + (2J+3)^{-1}\big)\rho) - (\Tr(\tfrac{J_i}{J+1}\rho))^2.
        \end{align}
        In a similar way, we obtain for $i\neq l$
        \begin{align}
            \cov(\omega_i,\omega_l) &= \Tr(\big(\tfrac{J_iJ_l+J_lJ_i}{(J+1)(2J+3)}\big)\rho) - \Tr(\tfrac{J_i}{J+1}\rho)\Tr(\tfrac{J_l}{J+1}\rho).
        \end{align}
        The state $\rho$ can be written as $\rho^{\otimes k}=c(xy)^{N_{AB}/2}q^{J_\mathbf{n}}$ with $x=M/(M+1)$, $y=N/(N+1)$, $q=x/y$ $c=(1-x)^k(1-y)^k$ where 
        \begin{align}
            \mathbf n=(\sin(2\theta)\cos(\phi),\sin(2\theta)\sin(\phi),\cos(2\theta)),
        \end{align}
        The operators $W(u):=\exp(\mathbbm i uJ_\mathbf{n})$ commute with $\rho$ and satisfy for every $T$ that commutes with $\rho$ and $J_x,J_y,J_z$ and for $\mathbf n^\perp$ satisfying $\langle\mathbf n,\mathbf n^\perp\rangle=0$ the equality
        \begin{align}
            \Tr(T\rho J_\mathbf{n^\perp} ) &=\frac{1}{2\pi}\int_0^{2\pi}\Tr(T\rho W(u)J_\mathbf{n}^\perp W(u)^\dagger )du= 0.
        \end{align}
        Thus for every $i,l$ we have
        \begin{align}
            \mathbb E(\omega_i) &= \Tr(\tfrac{J_i}{J+1}\rho) = n_i\Tr(\tfrac{J_\mathbf{n}}{J+1}\rho)\\
            \cov(\omega_i\omega_l) &= a\cdot n_in_l + b\cdot (\delta_{i,l}-n_in_l)
        \end{align}
        where
        \begin{align}
            a &= \Tr(\big(\tfrac{2J_\mathbf{n}^2}{(J+1)(2J+3)} + \tfrac{1}{2J+3}\big)\rho)-(\Tr(\tfrac{J_\mathbf{n}}{J+1}\rho))^2\\
            b &= \Tr(\big(\tfrac{J(J+1)-J_{\mathbf n}^{2}}{(J+1)(2J+3)} + \tfrac{1}{2J+3}\big)\rho).
        \end{align}
        where we used a decomposition of the spin operators $J_i$ as 
        \begin{align}
            J_i &= n_iJ_\mathbf{n} + J_{\mathbf{n}_i^\perp}.
        \end{align}
        Using the bounds of Lemmas \ref{lem:bound-on-Dz} to Lemma \ref{lem:bound-onJx2J-2}, this yields 
        \begin{align}\label{eqn:cov-omega}
            \cov(\hat\omega) = \frac{1}{k}\frac{2M(N+1)}{(M-N)^2}(\eins-|\mathbf{n}\rangle\langle\mathbf n|) + o(1/k)
        \end{align}

        \paragraph{Covariance of $(\hat\theta,\hat\phi)$}
        As before, we use the convention $\Delta<M-N$, $S:=2MN+M+N$. 

        The spin-coherent outcome is \(\omega\in S^2\), and the angular estimators are
        \begin{align}
            \hat\theta=\frac12\arccos(\tfrac{j+1}{j}\hat\omega_z), \qquad \hat\phi=\operatorname{atan2}(\hat\omega_y,\hat\omega_x).
        \end{align}
        From \eqref{eqn:cov-omega} we already have 
        \begin{align}
            \cov(\hat\omega) = \frac{1}{k}\frac{2M(N+1)}{(M-N)^2}(\eins-|\mathbf{n}\rangle\langle\mathbf n|) + o(1/k)
        \end{align}
        We define 
        \begin{align}
            J_\mathbf{n}:=\mathbf n\cdot\mathbf J, \qquad D:=J-J_\mathbf{n}.
        \end{align}
        We use the equality
        \begin{align}
            \|\omega-\mathbf{n}\|^2&=\langle\omega-\mathbf{n},\omega-\mathbf n\rangle\\
                &= \langle\omega,\omega\rangle-2\langle\omega,\mathbf n\rangle+1\\
                &= 2(1-\langle\omega,\mathbf n\rangle)
        \end{align}
        which holds since $\omega\in S^2$. Accordingly, we can write 
        \begin{align}
            \|\omega-\mathbf n\|^4=4(1-\langle \omega,\mathbf n\rangle)^2.
        \end{align}
        For any fixed value of $j$, we have by using \eqref{eqn:second-order-spin-moment} of Lemma \ref{lem:first-and-second-order-spin-momemnts}
        \begin{align}
            \int&(1-\langle\omega,\mathbf n\rangle)^2M_{j,t}(d\omega)
                = E_jE_t  + \int\sum_i\omega_in_i\left[- 2 + \omega_in_i + 2\sum_{l\neq i}\omega_ln_l\right]M_{j,t}(d\omega)\\
                &= E_jE_t  - 2\sum_in_i\frac{E_jE_tJ_i}{j+1} + \sum_{i\neq l}n_in_l\frac{E_jE_t\{J_iJ_l\}}{(2j+3)(j+1)}
                + \sum_in_i^2\left[\frac{2E_jE_tJ_i^2}{(j+1)(2j+3)}+\frac{E_jE_t}{2j+3}\right]\\
                &= \frac{2E_jE_t(D_\mathbf{n}+\eins)(D_\mathbf{n}+2\eins)}{(j+1)(2j+3)}
        \end{align}
        where we used $D_\mathbf{n}:=J-J_\mathbf{n}$. Averaging over $(j,t)$, noting that $J+c\geq J$ for $c>0$ and utilizing the bounds from Lemma \ref{lem:bound-onJx2J-2}, we arrive at 
        \begin{align}
            \mathbb E(\|\omega-\mathbf n\|^4) &=\mathcal O(k^{-2}).
        \end{align}
        Thus, Taylor expansion yields the desired 
        \begin{align}
            \cov(\hat\theta,\hat\phi) = \frac{2M(N+1)}{k(M-N)^2}\left(\begin{array}{ll}1/4&0\\0&\sin(2\theta)^{-2}\end{array}\right) + o(1/k).
        \end{align}
        The covariance between $\hat\omega$ and $\hat M$ as well as $\hat N$ vanishes at leading order, as can be seen by utilizing the tools derived for the proof of Theorem \ref{thm:restricted-case}.

        \paragraph{Holevo bound for the general model:}
                Once the parameter $\phi$ is treated as well, the \glspl{sld} for $M,N,\theta$ are updated to 
                \begin{align}
                    L_i\to V_A(\phi)L_iV_A(\phi)^\dagger.
                \end{align}
                The \gls{sld} for $\phi$ can be calculated to equal 
                \begin{align}
                    L_\phi = -\mathbbm{i}\frac{(N-M)\sin(2\theta)}{2MN+M+N}(e^{\mathbbm{i}\phi}a^\dagger b - e^{-\mathbbm{i}\phi}ab^ \dagger).
                \end{align}
                Following \cite[Eq. (6)]{ragy2016compatibility} we can compute the Holevo bound for the weight matrix $G=\mathrm{diag}(1,1,1,\sin^2(2\theta)/4)$ for this model as an optimization problem over Hermitian matrices $X_M, X_N, X_\theta, X_\phi$ as
                \begin{align}
                    H_G:=\min_{\mathbf X}\Tr(\Re\{G\cdot V(\mathbf X)\})+\|\sqrt{G}\cdot \Im\{V(\mathbf X)\}\cdot\sqrt{G}\|_1
                \end{align}
                where the minimization is over $\mathbf X = (X_M,X_N,X_\theta,X_\phi)$ that fulfill $\Tr(\{X_i,L_j\}\rho)=2\delta_{i,j}$ and $V(\mathbf X)_{ij}=\Tr(X_iX_j\rho)$. For the so-called $D$-invariant models, this bound coincides with the RLD CR bound 
                \begin{align}
                    C_G:=\Tr(G\cdot\Re\{\mathbf{L}^{-1}\}) + \|\sqrt{G}\cdot\Im\{\mathbf{L}^{-1}\}\cdot \sqrt{G}\|_1
                \end{align}
                (see \cite{holevo2011probabilistic}) where $\mathbf L_{ij}=\Tr(\hat L_j\hat L_i^\dagger\rho)$ for the RLD operators $\hat L_i$, which are defined as $\partial_i\rho=\rho\hat L_i$. A model is called $D$-invariant, if the super-operator $\mathcal D_\rho$ defined by $\rho\mathcal D_\rho(X)+\mathcal D_\rho(X)\rho=2\mathbbm{i}[X,\rho]$ \cite{holevo2011probabilistic} satisfies $\mathcal D_\rho(T_\rho)\subset T_\rho$ for the real vector space $T_\rho$ spanned by the \glspl{sld} $L_M,L_N,L_\theta,L_\phi$ (and for every choice of model parameters). The latter is true in the current model, where explicit calculation shows that 
                \begin{align}
                    \mathcal D_\rho(L_M)=\mathcal D_\rho(L_N)=0,\qquad \mathcal D_\rho(L_\theta)\propto L_\phi,\qquad \mathcal D_\rho(L_\phi)\propto L_\theta.
                \end{align}
                For the case at hand, we can therefore use $C_G$ instead of $H_G$. The operator $\mathbf L$ is given by 
                \begin{align}
                    \mathbf L^{-1}=\left(\begin{array}{cccc}
                        M(M+1)&0&0&0\\
                        0&N(N+1)&0&0\\
                        0&0&\tfrac{2MN+M+N}{4(M-N)^2}&-\tfrac{\mathbbm{i}}{2(M-N)\sin(2\theta)}\\
                        0&0&\tfrac{\mathbbm{i}}{2(M-N)\sin(2\theta)}&\tfrac{2MN+M+N}{(M-N)^2\sin(2\theta)^2}
                    \end{array}\right),
                \end{align}
                and directly gives the statement of Theorem \ref{thm:main}.
        \end{proof}

\subsection{Convergence analysis}
\label{sec:convergence_and_complexity_analysis}
\begin{proof}[Proof of Theorem \ref{thm_convergence_analysis}]
The proof has several steps. To simplify the terms, we start by defining the following notations:
\begin{equation}
   g_\ell:=\Delta-2s_\ell, \quad   x_\ell
    :=
    \frac{K^2}{\Delta}s_\ell,
    \qquad
    \beta
    :=
    \frac{K^3}{\Delta}\epsilon_{\mathrm{loc}}
\end{equation}

\paragraph{Monotonicity.}
First, we prove that the off-diagonal norm does not increase, so the iterate remains within the initial local basin. 
A Jacobi rotation preserves the Frobenius norm and unitarily mixes the
nonpivot off-diagonal entries. Hence
\begin{equation}
    s_{\ell,t+1}^2
    =
    s_{\ell,t}^2
    -
    2|z_{\ell,t}|^2
    +
    2|r_{\ell,t}|^2.
    \label{eq:pivot_balance}
\end{equation}
A skipped pivot leaves $s_{\ell,t}$ unchanged, whereas an executed pivot
satisfies $|r_{\ell,t}|\le\epsilon_{\mathrm{loc}} < |z_{\ell,t}|$. Thus
$s_{\ell,t}$ is nonincreasing, and
\begin{equation}
    s_{\ell+1} \le s_\ell \le s_0 < \frac{\Delta}{K^2}.
    \label{eq:pivot_bound_2}
\end{equation}
Therefore, $x_\ell < 1$, proving monotonicity, and under the assumption \textup{A2} gives an important inequality
\begin{equation}
    g_\ell > \Delta\left(1-\frac{2}{K^2}\right) = \frac{\Delta(K^2-2)}{K^2} > 0.
    \label{eq:gap_bound}
\end{equation}

\paragraph{Diagonal separation.} This is a technical step needed later to convert a small residual into a small angular error.
Let $D_{\ell,t}:=\operatorname{diag}(C_{\ell,t})$. Since $C_{\ell,t}$ is
unitarily similar to $C^{(0)}$, the Hoffman--Wielandt theorem \cite{HoffmanWielandt1953} gives a
permutation $\pi_t$ such that
\begin{equation}
    \sum_{i=1}^K
    \left|
        (C_{\ell,t})_{ii}
        -
        \lambda_{\pi_t(i)}
    \right|^2
    \le
    \|C_{\ell,t}-D_{\ell,t}\|_{\mathrm F}^2
    =
    s_{\ell,t}^2.
\end{equation}
Consequently,
$|(C_{\ell,t})_{ii}-\lambda_{\pi_t(i)}|\le s_\ell$, and hence
\begin{equation}
    |(C_{\ell,t})_{ii}-(C_{\ell,t})_{jj}|
    \ge
    \Delta-2s_\ell
    =
    g_\ell,
    \qquad i\neq j.
    \label{eq:diagonal_separation}
\end{equation}
For a phase-aligned pivot block
$\left(\begin{smallmatrix}a&z\\z&d\end{smallmatrix}\right)$, its two
eigenvalues are separated by
\begin{equation}
    \rho_{\ell,t}
    :=
    \sqrt{(d-a)^2+4|z|^2}
    \ge
    |d-a|
    \ge
    g_\ell.
    \label{eq:pivot_gap}
\end{equation}

\paragraph{Angular error.} Now, using the diagonal separation argument, we show how residual threshold controls the rotation error.
This angle bound is what later makes the regeneration error quadratic.

For an executed pivot, let $\theta^*$ be the exact Jacobi angle and write
$\widehat\theta_{\ell,t}=\theta^*+\delta\theta$. Since
$|z_{\ell,t}|\le s_\ell/\sqrt2$, one has
\begin{equation}
    |\theta^*| \le \frac{|z_{\ell,t}|}{|d-a|} \le \frac{s_\ell}{\sqrt2\,g_\ell} < \frac{1}{\sqrt2(K^2-2)} < \frac{\pi}{12}.
\end{equation}
where we used Eq. \ref{eq:pivot_bound_2} and \ref{eq:gap_bound} in third inequality.
Because $|\widehat\theta_{\ell,t}|\le\pi/4$, it follows that
$|\delta\theta| < \pi/3$. The pivot residual satisfies
\begin{equation}
    |r_{\ell,t}|
    =
    \frac{\rho_{\ell,t}}{2}
    |\sin(2\delta\theta)|.
    \label{eq:residual_angle}
\end{equation}
If $|\delta\theta| > \pi/4$, then $\pi/2 < 2|\delta\theta| < 2\pi/3$, so
$|\sin(2\delta\theta)|\ge\sqrt3/2$. On the other hand,
\begin{equation}
    |\sin(2\delta\theta)| = \frac{2|r_{\ell,t}|}{\rho_{\ell,t}} \le \frac{2\epsilon_{\mathrm{loc}}}{g_\ell} < \frac{1}{3K(K^2-2)} < \frac{\sqrt3}{2},
\end{equation}
which is a contradiction. Hence $|\delta\theta|\le\pi/4$. Since
$|\sin(2u)|\ge(4/\pi)|u|$ on this interval,
\begin{equation}
    |\delta\theta|
    \le
    \frac{\pi|r_{\ell,t}|}{2\rho_{\ell,t}}
    \le
    \frac{\pi\epsilon_{\mathrm{loc}}}{2g_\ell}
    =:
    \eta_\ell.
    \label{eq:angle_error}
\end{equation}

\paragraph{One-sweep estimate.} This step converts local pivot control into a global one-sweep recurrence. This is crucial as later rotations can regenerate the pivot even though it has been handled during its Jacobi rotation. Under the assumptions \textup{A1-A2}, we now show that quadratic convergence recurrence is possible.

Consider a later rotation on $(i,k)$. Its implemented angle satisfies
\begin{equation}
    |\sin\widehat\theta|
    \le
    \frac{|c_{ik}|}{g_\ell}
    +
    \eta_\ell,
\end{equation}
where $c_{pq}:=(C_{\ell,t})_{pq}$. If $(i,j)$ has already been processed,
then $|c_{ij}|\le\epsilon_{\mathrm{loc}}$, and the later rotation gives
\begin{equation}
    c_{ij}^{+}
    =
    \cos\widehat\theta\,c_{ij}
    +
    \omega\sin\widehat\theta\,c_{kj},
    \qquad
    |\omega|=1.
\end{equation}
Therefore,
\begin{equation}
    |c_{ij}^{+}|-|c_{ij}|
    \le
    \frac{|c_{ik}||c_{kj}|}{g_\ell}
    +
    \eta_\ell|c_{kj}|.
\end{equation}
Since $2|c_{ik}|^2+2|c_{kj}|^2\le s_{\ell,t}^2\le s_\ell^2$,
\begin{equation}
    |c_{ij}^{+}|-|c_{ij}|
    \le
    \frac{s_\ell^2}{4g_\ell}
    +
    \frac{\eta_\ell s_\ell}{\sqrt2}.
\end{equation}
At most $2(K-2)$ later rotations can affect $(i,j)$. Thus
\begin{equation}
    |C_{ij}^{(\ell+1)}|
    \le
    \epsilon_{\mathrm{loc}}
    +
    \frac{(K-2)s_\ell^2}{2g_\ell}
    +
    \sqrt2(K-2)\eta_\ell s_\ell.
\end{equation}
Taking the off-diagonal Frobenius norm and using
$\sqrt{K(K-1)}\le K$ yields
\begin{equation}
    s_{\ell+1}
    \le
    K\epsilon_{\mathrm{loc}}
    +
    \frac{K(K-2)}{2g_\ell}s_\ell^2
    +
    \sqrt2K(K-2)\eta_\ell s_\ell.
    \label{eq:sweep_bound}
\end{equation}
Multiplying by $K^2/\Delta$ gives
\begin{equation}
    x_{\ell+1}
    \le
    \beta
    +
    q_\ell x_\ell^2
    +
    \gamma_\ell x_\ell,
\end{equation}
where
\begin{equation}
    q_\ell
    :=
    \frac{(K-2)\Delta}{2Kg_\ell},
    \qquad
    \gamma_\ell
    :=
    \sqrt2K(K-2)\eta_\ell.
\end{equation}
By \eqref{eq:gap_bound},
\begin{equation}
    q_\ell < \frac{K(K-2)}{2(K^2-2)} \le \frac12.
\end{equation}
Moreover, by \eqref{eq:gap_bound} and \eqref{eq:angle_error},
\begin{equation}
    \frac{\gamma_\ell}{\beta/2} < \frac{\sqrt2\pi(K-2)}{K^2-2} \le \frac{\sqrt2\pi}{7} < 1,
\end{equation}
where $(K-2)/(K^2-2)\le1/7$ for every integer $K\ge3$. Hence
$\gamma_\ell\le\beta/2$. Since $\gamma_\ell\le\beta/2$ and $x_\ell<1$, one has
$\gamma_\ell x_\ell\le\beta/2$. Therefore,
\begin{equation}
    x_{\ell+1}
    \le
    \frac{3}{2}\beta+\frac{1}{2}x_\ell^2.
    \label{eq:normalized_recurrence}
\end{equation}
Whenever $x_\ell\ge\sqrt{3\beta}$,
\[
    \frac{3}{2}\beta
    \le
    \frac{1}{2}x_\ell^2,
\]
and hence $x_{\ell+1}\le x_\ell^2$. Once
$x_\ell\le\sqrt{3\beta}$, one additional sweep gives
$x_{\ell+1}\le3\beta$.

\paragraph{Sweep count and global stopping.} Finally, we compute the run-time of the algorithm.
Since $\beta\le1/6$ and $x_0<1$,
\begin{equation}
    x_1
    \le
    \frac{3}{2}\beta+\frac{1}{2}x_0^2
    <
    \frac{3}{4}.
\end{equation}
During the quadratic-contraction regime,
\begin{equation}
    x_{1+r}
    \le
    \left(\frac{3}{4}\right)^{2^r}.
\end{equation}
Thus $x_\ell\le\sqrt{3\beta}$ once
\begin{equation}
    2^r
    \ge
    \frac{
        \log\!\bigl(1/\sqrt{3\beta}\bigr)
    }{
        \log(4/3)
    }.
\end{equation}
Consequently,
\begin{equation}
    L
    =
    O\!\left(
        \log\log\frac{1}{\beta}
    \right)
    =
    O\!\left(
        \log\log
        \frac{\Delta}{K^3\epsilon_{\mathrm{loc}}}
    \right).
\end{equation}

\end{proof}

\begin{proof}[Proof of Theorem \ref{thm:sample_complexity}]
Assume that the failure probability of a single unit-cell invocation
using $k$ copies satisfies
\begin{equation}
    P_{\mathrm{cell}}(k,\epsilon_{\mathrm{loc}})
    \leq
    \exp\!\left[
        -kC_{\star}\epsilon_{\mathrm{loc}}^{2}
        +r(k)
    \right],
    \qquad
    r(k)=O(\log k).
\end{equation}
By the definition of $O(\log k)$, there exist constants
$\beta\geq0$ and $k_{0}\geq1$ such that $ |r(k)|
    \leq
    \beta\log k,$ for every $ k\geq k_{0}.$ Therefore,
\begin{align}
    P_{\mathrm{cell}}(k,\epsilon_{\mathrm{loc}})
    &\leq
    \exp\!\left[
        -kC_{\star}\epsilon_{\mathrm{loc}}^{2}
        +\beta\log k
    \right]
    \nonumber\\
    &=
    \exp\!\left(
        -kC_{\star}\epsilon_{\mathrm{loc}}^{2}
    \right)
    \exp(\beta\log k)
    \nonumber\\
    &=
    k^{\beta}
    \exp\!\left(
        -kC_{\star}\epsilon_{\mathrm{loc}}^{2}
    \right),
    \qquad
    k\geq k_{0}.
\end{align}
The condition
$P_{\mathrm{cell}}(k,\epsilon_{\mathrm{loc}})
\leq\alpha_{\mathrm{loc}}$
is guaranteed whenever
\begin{equation}
    kC_{\star}\epsilon_{\mathrm{loc}}^{2}
    -
    \beta\log k
    \geq
    \log\!\left(
        \frac{1}{\alpha_{\mathrm{loc}}}
    \right).
\end{equation}
A sufficient asymptotic choice is thus
\begin{equation}
    k
    =
    O\!\left[
        \frac{1}
             {C_{\star}\epsilon_{\mathrm{loc}}^{2}}
        \left\{
            \log\!\left(
                \frac{1}{\alpha_{\mathrm{loc}}}
            \right)
            + D
        \right\}
    \right], \quad D = \beta
            \log\!\left(
                1+
                \frac{1}
                     {C_{\star}\epsilon_{\mathrm{loc}}^{2}}
                \log\!\left(
                    \frac{1}{\alpha_{\mathrm{loc}}}
                \right)
            \right)
\end{equation}

Since the network contains
$O\!\left(KL(\epsilon_{\mathrm{loc}})\right)$
unit-cell invocations, the total number of copies satisfies $  N(\epsilon_{\mathrm{loc}},\alpha_{\mathrm{loc}})
    =
    O\!\left[
        K L(\epsilon_{\mathrm{loc}})\,k
    \right].$
Using $L(\epsilon_{\mathrm{loc}})
    =
    O\!\left(
        \log\log
        \frac{1}{\epsilon_{\mathrm{loc}}}
    \right),$
we obtain
\begin{align}
    N(\epsilon_{\mathrm{loc}},\alpha_{\mathrm{loc}})
    =
    O\!\Bigg[
        \frac{K}
             {C_{\star}\epsilon_{\mathrm{loc}}^{2}}
        \log\log\!\left(
            \frac{1}{\epsilon_{\mathrm{loc}}}
        \right)
        \Bigg\{
            &
            \log\!\left(
                \frac{1}{\alpha_{\mathrm{loc}}}
            \right) + D
        \Bigg\}
    \Bigg].
\end{align}
The first term gives the leading exponential-order scaling, while the second is the subleading logarithmic correction induced by the
$O(\log k)$ term.

It remains to control the failure probability of the complete network. Let
$\mathcal{F}_r$ denote the failure event of the $r$-th unit-cell
invocation. Since the total number of unit-cell invocations is at most $R
    =
    \frac{K L(\epsilon_{\mathrm{loc}})}{2},$
and $\Pr(\mathcal{F}_r)\leq\alpha_{\mathrm{loc}}$ for every $r$, the union
bound yields
\begin{align}
    \Pr\!\left(
        \bigcup_{r=1}^{R}\mathcal{F}_r
    \right)
    &\leq
    \sum_{r=1}^{R}\Pr(\mathcal{F}_r)
    \nonumber\\
    &\leq
    \frac{K L(\epsilon_{\mathrm{loc}})}{2}
    \alpha_{\mathrm{loc}}.
\end{align}
Consequently, the complete network succeeds with probability at least $1-
    \frac{K L(\epsilon_{\mathrm{loc}})}{2}
    \alpha_{\mathrm{loc}}.$

\end{proof}

\subsection{Quantum advantage}
\label{sec:quantum_advantage}
We now compare the copy complexity of the sensing-based approach with that of a conventional heterodyne tomography strategy. Similar analysis has been performed previously in the context of weak-Thermal-Light Interferometry and Brillouin sensing \cite{PhysRevLett.117.190801, PhysRevLett.107.270402,adhikari2026fundamentallimitseavesdropperdetection}. 

\paragraph{Fisher-information comparison}

As our output state, $\rho$ is a Gaussian state, it is useful to write it in the covariance-matrix representation \cite{bookserfini,HolevoBook}. We order the quadratures as $ \mathbf r=(x_a,x_b,p_a,p_b)^T $ where, $x_a=a+a^\dagger$, $p_a=(a-a^\dagger)/i,  x_b=b+b^\dagger, p_b=(b-b^\dagger)/i$ such that $[x_j,p_k]=2i\delta_{jk}.$ The Gaussian state of interest to us has zero mean with covariance matrix $\Sigma_{ij}
    =
    \frac{1}{2}
    \langle r_i r_j+r_j r_i\rangle .$ Therefore, the input covariance matrix of $S_M\otimes S_N$ is
\begin{equation}
    \Sigma_{\rm in}
    =
    \begin{pmatrix}
        2M+1 & 0 & 0 & 0\\
        0 & 2N+1 & 0 & 0\\
        0 & 0 & 2M+1 & 0\\
        0 & 0 & 0 & 2N+1
    \end{pmatrix}.
\end{equation}
The output covariance matrix has the block form $ \Sigma(M,N,\theta)
    =
    \Sigma_x(M,N,\theta)\oplus \Sigma_p(M,N,\theta),$

\begin{equation}
    \Sigma_x(M,N,\theta)
    =
    \Sigma_p(M,N,\theta)
    =
    R(\theta)
    \begin{pmatrix}
        2M+1 & 0\\
        0 & 2N+1
    \end{pmatrix} R(\theta)^T
\end{equation}
where, $ R(\theta)
    =
    \begin{pmatrix}
        \cos\theta & \sin\theta\\
        -\sin\theta & \cos\theta
    \end{pmatrix}$ is supposed to represent the beamsplitter. Explicitly, 
\begin{equation}
    \Sigma(M,N,\theta)
    =
    \begin{pmatrix}
        \tilde A & \tilde C & 0 & 0\\
        \tilde C & \tilde B & 0 & 0\\
        0 & 0 & \tilde A & \tilde C\\
        0 & 0 &\tilde C &\tilde B
    \end{pmatrix},
    \label{eq:covariance_matrix_output}
\end{equation}
with $ \tilde A = 2(M\cos^2\theta + N \sin^2 \theta)+1
     = 2\langle a^\dagger a\rangle+1,$ $ \tilde B = 2(M\sin^2\theta + N\cos^2\theta)+1 = 2\langle b^\dagger b\rangle+1,$ and $ \tilde C
    =
    (N-M)\sin2\theta .$

The ultimate limit for the Gaussian covariance matrix is given by the quantum Fisher information (QFI) \cite{monras2013phase,_afr_nek_2018}:
\begin{equation}
    F^{\rm Q}_{\mu\nu}
    =
    \tfrac12\,
    \mathrm{vec}[\partial_\mu\Sigma]^{T}
    \big(\Sigma\otimes\Sigma-\Omega\otimes\Omega\big)^{-1}
    \mathrm{vec}[\partial_\nu\Sigma],
    \qquad
    \Omega=\begin{pmatrix}0 & \eins_2\\ -\eins_2 & 0\end{pmatrix},
    \label{eq:gaussian_qfi}
\end{equation}
Evaluating \eqref{eq:gaussian_qfi} gives a diagonal,
$\theta$-independent matrix,
\begin{equation}
    F^{\rm Q}
    =
    \mathrm{diag}\!\left(
        \frac{1}{M(M+1)},\;
        \frac{1}{N(N+1)},\;
        \frac{4(M-N)^2}{2MN+M+N}
    \right),
    \label{eq:qfi_matrix}
\end{equation}
with corresponding quantum Cram\'er--Rao bounds
\begin{equation}
    \mathrm{Var}(\hat M)\ge\frac{M(M+1)}{k},
    \quad
    \mathrm{Var}(\hat N)\ge\frac{N(N+1)}{k},
    \quad
    \mathrm{Var}(\hat\theta)\ge\frac{2MN+M+N}{4k\,(M-N)^2}.
    \label{eq:qcrb}
\end{equation}
It is possible that this bound cannot be achieved by any measurements, and one has to consider the Holevo bound. Still, Eq.\ref{eq:qcrb} provides us with an upper baseline. 

Now, let us compare the Fisher information obtained from a state-of-the-art heterodyne measurement with the quantum Cramer-Rao bound. This gives us a baseline for the comparison with our proposed protocol.  Heterodyning the output state $\rho(M,N,\theta)$ gives an outcome distribution with covariance matrix $\Sigma_{\rm het}$ as
\begin{equation}
    \Sigma_{\rm het}
    =
    \Sigma(M,N,\theta)+ \eins_4
    =
    \begin{pmatrix}
        \tilde A+1 & \tilde C & 0 & 0\\
        \tilde C & \tilde B+1 & 0 & 0\\
        0 & 0 & \tilde A+1 & \tilde C\\
        0 & 0 & \tilde C & \tilde B+1
    \end{pmatrix}
    \;\equiv\;
    W\oplus W,
    \label{eq:het_covariance}
\end{equation}
where, $ W=\begin{pmatrix}\tilde A+1 & \tilde C\\[2pt] \tilde C & \tilde B+1\end{pmatrix}$. The heterodyne Fisher information matrix is
\begin{equation}
    F^{\rm het}_{\mu\nu}
    =
    \tfrac12\,
    \mathrm{Tr}\!\left[
        \Sigma_{\rm het}^{-1}(\partial_\mu\Sigma_{\rm het})\,
        \Sigma_{\rm het}^{-1}(\partial_\nu\Sigma_{\rm het})
    \right] =  \mathrm{Tr}\!\left[
        W^{-1}(\partial_\mu W)\,W^{-1}(\partial_\nu W)
    \right].
    \label{eq:gaussian_cfi}
\end{equation}
Explicitly, the heterodyne Fisher information is diagonal and
independent of $\theta$,
\begin{equation}
    F^{\rm het}
    =
    \mathrm{diag}\!\left(
        \frac{1}{(M+1)^2},\;
        \frac{1}{(N+1)^2},\;
        \frac{2(M-N)^2}{(M+1)(N+1)}
    \right).
    \label{eq:het_fisher}
\end{equation}
For $k$ independent copies, the heterodyne variance scales as:
\begin{equation}
    \mathrm{Var}(\hat M)_{\rm het}\ge\frac{(M+1)^2}{k},
    \quad
    \mathrm{Var}(\hat N)_{\rm het}\ge\frac{(N+1)^2}{k},
    \quad
    \mathrm{Var}(\hat\theta)_{\rm het}\ge\frac{(M+1)(N+1)}{2k\,(M-N)^2}.
    \label{eq:het_crb}
\end{equation}

Comparing  Eq. \eqref{eq:het_fisher} with quantum limit Eq. \eqref{eq:qfi_matrix}, heterodyne is
strictly suboptimal for every parameter,
\begin{equation}
    \frac{F^{\rm het}_{MM}}{F^{\rm Q}_{MM}}=\frac{M}{M+1},
    \qquad
    \frac{F^{\rm het}_{NN}}{F^{\rm Q}_{NN}}=\frac{N}{N+1},
    \qquad
    \frac{F^{\rm het}_{\theta\theta}}{F^{\rm Q}_{\theta\theta}}
    =
    \frac{2MN+M+N}{2(M+1)(N+1)}<1 .
    \label{eq:ratios}
\end{equation}
Figure~\ref{fig:fisher_ratio_comparison} illustrates $F^{\mathrm{het}}_{\theta\theta}/F^{Q}_{\theta\theta}$ part of  Eq.~\eqref{eq:ratios} over the full range of thermal occupations. The ratio is strictly below unity everywhere on the $(M,N)$ plane, confirming that heterodyne detection is never optimal for estimating the beam-splitter angle
at any finite occupation. Heterodyne detection approaches the quantum
    limit (ratio $\to 1$) only when both modes are bright, while the quantum advantage (ratio $\to 0$) is largest when either occupation is small.

\begin{figure}
        \centering
        \includegraphics[ width=\linewidth]{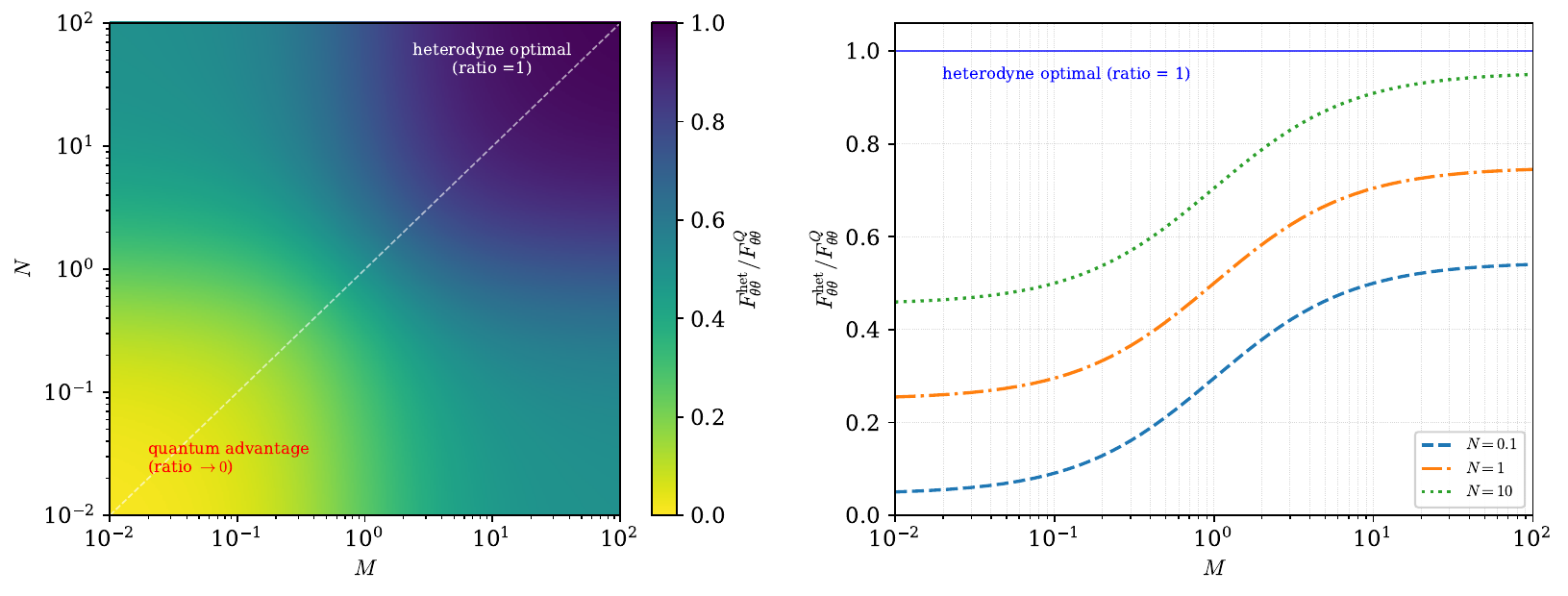}
         \caption{Comparison of the heterodyne and quantum-optimal Fisher information
    for the beam-splitter angle $\theta$, given by the ratio
    $F^{\mathrm{het}}_{\theta\theta}/F^{Q}_{\theta\theta}$ in Eq.~\eqref{eq:ratios}. 
    The left panel shows that the ratio is less than unity for all finite thermal occupations $M,N$. The right panel includes slices at some fixed $N$. }
        \label{fig:fisher_ratio_comparison}
\end{figure}

\paragraph{Unit-cell comparison}

\begin{proof}[Proof of Theorem \ref{thm:unit_cell_performance}] We now treat the different unit cell mechanisms:
\paragraph{Spin \gls{povm}}
In this part of our proof, we make use of representation-theoretic tools. Their use in quantum information theory is documented as early as 1982 \cite{Holevo1982Probabilistic}, later work studied spectrum estimation and entropic bounds \cite{KeylWerner2001Spectrum,Hayashi2001RelativeEntropy}, followed by applications to problems concerning the structure of multipartite states, spectral compatibility, entanglement, quantum method of types, and, more recently, symmetry-reduced quantum information processing and mixed Schur--Weyl duality \cite{ChristandlMitchison2006Spectra,ChristandlHarrowMitchison2007Kronecker,WalterDoranGrossChristandl2013EntanglementPolytopes,Noetzel2014HypothesisTesting,GrossNezamiWalter2021CliffordSchurWeyl,GrinkoOzols2024LinearProgramming,Nguyen2023MixedSchur}. In the present setting, we apply representation theory to estimate the multiplicity $\mu_{j,t}$ of the representation of the estimator $J_\omega$ inside the Hilbert space $\mathrm{range}(E_jE_t)\subset\mathcal F^{\otimes k}$. 

Consider the measurement of $(M,N,\theta,\phi)$ where $(\theta,\phi)$ are estimated as described in \eqref{def:(theta,phi)-estimator}. 
The state $\rho$ can be written as 
\begin{align}
    \rho^{\otimes k}&=U(\theta,\phi)^{\otimes k}c(xy)^{N_{AB}/2}q^{J_\mathbf{z}}U(\theta,\phi)^{\dagger\otimes k}\\
        &=c(xy)^{N_{AB}/2}q^{J_\mathbf{n}}
\end{align} 
with $x=M/(M+1)$, $y=N/(N+1)$, $q=x/y$, (implying $q>1$) and $c=(1-x)^k(1-y)^k$, where 
\begin{align}
    \mathbf n=(\sin(2\theta)\cos(\phi),\sin(2\theta)\sin(\phi),\cos(2\theta)).
\end{align}
We emphasize the conditions $[N_{AB},U(\theta,\phi)^{\otimes k} ]=[J,U(\theta,\phi)^{\otimes k}]=0$. After obtaining measurement result $t$ of the total photon number $N_{AB}$, and $j$ of the spin operator $J$, the state is thus given by 
\begin{align}
    p(j,t)^{-1}\cdot E_jE_t\rho^{\otimes k} E_jE_t
        &=p(j,t)^{-1}\cdot c\cdot (xy)^{t/2}\cdot E_jE_tq^{J_\mathbf{n}}E_jE_t,
\end{align}
where $p(j,t)$ is the (nonzero) probability of obtaining result $(j,t)$. Thus, a follow-up measurement of $\omega$ has probability density
\begin{align}
    p(\omega|j,t) 
        &= \tfrac{c\cdot (xy)^{t/2}}{p(j,t)}\cdot \tfrac{2j+1}{4\pi}\cdot\Tr(M_j(\omega)E_t q^{J_\mathbf{n}})\\
        &= \tfrac{c\cdot (xy)^{t/2}}{p(j,t)}\cdot \mu_{j,t}\cdot \tfrac{2j+1}{4\pi}\cdot \langle j,\omega|q^{J_\mathbf{n}}|j,\omega\rangle\\
        &= \tfrac{c\cdot (xy)^{t/2}}{p(j,t)}\cdot \mu_{j,t}\cdot \tfrac{2j+1}{4\pi}\cdot \langle j,\omega|e^{\ln(q)J_\mathbf{n}}|j,\omega\rangle\\
        &= \tfrac{c\cdot (xy)^{t/2}}{p(j,t)}\cdot \mu_{j,t}\cdot \tfrac{2j+1}{4\pi}\sum_{m=0}^{2j}q^{j-m}\langle j,\omega|j,\mathbf n,j-m\rangle\langle j,\mathbf n,j-m|j,\omega\rangle\\
        &= \tfrac{c\cdot (xy)^{t/2}}{p(j,t)}\cdot \mu_{j,t}\cdot \tfrac{2j+1}{4\pi}q^j\sum_{m=0}^{2j}\binom{2j}{m}\big(\cos(\nu/2)^2\big)^{2j-m}\big(\sin(\nu/2)^{2}/q\big)^m\\
        &= \tfrac{c\cdot (xy)^{t/2}}{p(j,t)}\cdot \mu_{j,t}\cdot \tfrac{2j+1}{4\pi}q^j(\cos(\nu/2)^2+\sin(\nu/2)^{2}q^{-1})^{2j}
\end{align}
where we used \cite[Eqn. (A.1)]{Manai_2023}: Letting $\nu$ be the angle between $\omega$ and $\mathbf n$, it holds
\begin{align}
    |\langle j,\omega,j-m|j,\mathbf n,j\rangle|^2
        &=\binom{2j}{m}\cos(\nu/2)^{4j-2m}\sin(\nu/2)^{2m}.
\end{align}
The probability that $|\nu|\geq\epsilon$ can be derived via the following integral:
\begin{align}
    \int_{\epsilon}^{\pi}q^j(\cos(\nu/2)^2+\sin(\nu/2)^{2}q^{-1})^{2j}&\sin(\nu)d\nu
        =     q^{-j}\int_{\epsilon}^{\pi}\big(q - (q-1)\sin(\nu/2)^{2}\big)^{2j}\sin(\nu)d\nu\\
        &= -q^{-j}\tfrac{2}{(2j+1)(q-1)}\Big[\big[q - (q-1)\sin(\nu/2)^{2}\big]^{2j+1}\Big]^{\pi}_\epsilon\\
        &= \frac{2q^{-j}}{(2j+1)(q-1)}\Big[\big((q- (q-1) \sin^2(\epsilon/2)\big)^{2j+1}-1\Big].
\end{align}
With this integral, we obtain 
\begin{align}
    \mathbb P(|\nu|\geq\epsilon|j,t) 
        &= \tfrac{c\cdot (xy)^{t/2}}{p(j,t)}\cdot \mu_{j,t}\cdot \tfrac{2j+1}{2}\frac{2q^{-j}}{(2j+1)(q-1)}\Big[\big((q- (q-1) \sin^2(\epsilon/2)\big)^{2j+1}-1\Big].\label{eqn:p(nu>eps)}
\end{align}
Since $\mathbb P(|\nu|\geq0|j,t)=1$ must hold and since $\nu\geq0$, it follows 
\begin{align}
    \mathbb P(\nu\geq\epsilon|j,t) 
        &= \frac{1}{q^{2j+1}-1}\Big[\big((q- (q-1) \sin^2(\epsilon/2)\big)^{2j+1}-1\Big]\\
        &\leq \big((1- (1-1/q) \sin^2(\epsilon/2)\big)^{2j+1}\\
        &=\exp((2j+1)\log[1-(1-1/q)\sin^2(\epsilon/2)]).\label{eqn:nu_geq_eps|j,t}
\end{align}
By the relation $\|\omega - \mathbf n\|^2=2-2\langle\omega,\mathbf n\rangle = 4\sin^2(\nu/2)$ we thus obtained a bound on our precision in the estimate $\omega$ of $\mathbf n$. Before we can progress, we need to first compute $\mu_{j,t}$ and establish asymptotic bounds on it. 

The symmetrized $2k$-mode Hilbert space $\mathcal H_1\subset (\mathcal F\otimes\mathcal F)^{\otimes k}$ that contains exactly $1$ photon is isomorphic to $\mathbb C^2\otimes\mathbb C^k$, where the first factor $\mathbb C^2=\mathrm{span}(\{e_A,e_B\})$ identifies the occupancy in the $A$ and $B$ system and the second factor the copy of the system which contains the photon. Correspondingly, the Hilbert space $\mathcal H_t$ is isomorphic to $\mathrm{Sym}^t(\mathbb C^2\otimes\mathbb C^k)$ with an explicit mapping given by
\begin{align}
    |n_A^k,n_B^k\rangle \to \sqrt{\binom{t}{n_{A,1},\ldots,n_{A,k},n_{B,1},\ldots,n_{B,k}}}\mathbf{S}_t\bigotimes_{i=1}^k(e_A\otimes e_i)^{\otimes n_{A,i}}\otimes (e_B\otimes e_i)^{\otimes n_{B,i}}
\end{align}
where $n_A=(n_{A,1},\ldots,n_{A,k})$ and $n_B=(n_{B,1},\ldots,n_{B,k})$ are the photon numbers in the different modes of the $A$- and  the $B$ systems and the symmetrizer $\mathbf S_t=\tfrac{1}{t!}\sum_{\Pi\in \mathcal S_t}\Pi$ acts on the $t=\sum_in_{A,i}+n_{B,i}$ copies of $(\mathbb C^2\otimes\mathbb C^k)$ with $\mathcal S_t$ being the symmetric group. 
According to \cite[Exercise 6.11]{FultonHarris2004}, we can write 
\begin{align}
    \mathrm{Sym}^t(\mathbb C^2\otimes\mathbb C^k)&\simeq\bigoplus_\lambda S_\lambda(\mathbb C^2)\otimes S_\lambda(\mathbb C^k)\\
        &=\bigoplus_{s=0}^{t/2} S_{(t-s,s)}(\mathbb C^2)\otimes S_{(t-s,s)}(\mathbb C^k)
\end{align}
where the Young tableaux $\lambda=(\lambda_1,\lambda_2)$ are restricted to have a total of $t$ entries, with at most $k$ columns and $2$ rows of lengths $\lambda_1,\lambda_2$. Here, $S_\lambda(V)$ is the corresponding irreducible representation of the general linear group on $V$. Every such diagram can then be identified with a pair $(t-s,s)$ of natural numbers where $s\in\{0,\ldots,t/2\}$. 

The spin operators $J_x,J_y,J_z$ can, on $\mathrm{Sym}^t(\mathbb C^2\otimes\mathbb C^k)$, be represented as
\begin{align}
    \tilde J_x = \frac{1}{2}\sum_{i=1}^t\eins\otimes\ldots\otimes \tilde J_x\otimes\ldots\eins
\end{align}
where $\tilde J_x = \sigma_x\otimes\eins_k$. On $\bigoplus_{s=0}^{t/2} S_{(t-s,s)}(\mathbb C^2)\otimes S_{(t-s,s)}(\mathbb C^k)$, their action is restricted to $S_{(t-s,s)}(\mathbb C^2)$, and the same applies to $\mathcal J:=\sum_iJ_i^2$ \cite[Eq. (16)]{Yadin_2023}.
According to \cite[Theorem 6.4]{FultonHarris2004} we have $\mathrm{dim}(S_\lambda(\mathbb C^k)) = \prod_{1\leq i < j\leq k}\frac{\lambda_i-\lambda_j+j-i}{j-i}$. For any given diagram of shape $\lambda_1=t-s$, $\lambda_2=s$ we therefore have
\begin{align}
    \mathrm{dim}(S_\lambda(\mathbb C^k)) &= \frac{t-2s+1}{t-s+1}\binom{t-s+k-1}{k-1}\binom{s+k-2}{k-2}\label{eqn:combinatorial-dim(S_lambda)}\\
        &=\frac{2j+1}{t/2+j+1} \binom{t/2+j+k-1}{k-1}\binom{t/2-j+k-2}{k-2}\\
        &=\frac{a_k}{k\pi u_k\sqrt{u_kv_k(1+u_k)(1+v_k)^3}}\exp(k\cdot[h(u_k)+h(v_k)]+\mathcal O(\log k))\label{eqn:dim(S_lambda)}
\end{align}
where $u_k=\tfrac{t+2j}{2k}$ and $v_k=\tfrac{t-2j}{2k}$ are the estimators for $M$ and $N$ while $h(x):=(x+1)\log(x+1)-x\log x$ is the entropy of a mean zero thermal state with mean photon number $x$. The probability of obtaining measurement outcome $(j,t)$ can now be estimated tightly by inserting \eqref{eqn:dim(S_lambda)} into \eqref{eqn:p(nu>eps)} with $\epsilon=0$, giving
\begin{align}
    p(j,t) = c(xy)^{t/2}\mu_{j,t}\sum_{m=-j}^jq^m = c(xy)^{t/2}\mu_{j,t}q^{-j}\sum_{d=0}^{2j}q^d.
\end{align}
Upon evaluating the geometric series and writing out $c$ explicitly, we obtain 
\begin{align}
    p(j,t) \leq (1-x)^k(1-y)^k(xy)^{t/2}q^{-j}\cdot \frac{1-q^{2j+1}}{1-q}\cdot \exp(k\cdot[h(u_k)+h(v_k)]+\mathcal O(\log k)).
\end{align}

Using the relative entropy $D(R\|S)=R\log(R/S)-(R+1)\log[(R+1)/(S+1)]$ between two mean zero thermal states, and using the definition of $x$ and $y$, we therefore obtain
\begin{align}
    p(j,t)\leq \exp(-k\cdot [D(u_k\|M) + D(v_k\|N)]+\mathcal{O}(\log k)).
\end{align}
Since $D(R\|S)\geq0$ and $D(R\|S)=0$ if and only if $R=S$, $p(j,t)$ is suppressed for all realizations of $(j,t)$ for which the estimators $(u_k,v_k)$ deviates from the ``true'' parameters $(M,N)$.

Assume now that a selection rule was applied after the $(j,t)$ measurement, and that the $\omega$ measurement is only carried out if $j\in \mathbf A$ where $\mathbf A=\{j\geq k\cdot(\Delta/2-\delta)\}$ and $t\in B$ where $\mathbf B=\{t\leq k\cdot(2B+\Delta)\}$ for some $\delta>0$, $\Delta<M-N$ and $M<B$. While the first condition allows to decide whether a new rotation angles should be estimated and implemented, the second condition is mostly important for the purpose of obtaining sharp bounds in the following analysis: The probability for the event $\mathbf A\cap \mathbf B$ is bounded by first noting that 
\begin{align}
    \mathbb P(\mathbf B^\complement)
        &\leq \exp(-k\cdot2\cdot D(B+\Delta/2\|B))\\
        &\leq \exp\left(-k\tfrac{\Delta^2}{4(B+\Delta/2)(B+\Delta/2+1)}\right), 
\end{align}
(see \cite{dembo1998large} for the large deviation bound, the second bound follows from Taylor series expansion of $\log(x)$ around $x=1$). Conditioned on having a value $t\in \mathbf B$, it follows  
\begin{align}
    \mathbb P(\mathbf A^\complement\cap\mathbf B)
        &\leq \sum_{j\in \mathbf A^\complement}\sum_{t\in \mathbf B}\exp(-k[D(u_k\|M) + D(v_k\|N)] )\\
        &\leq k^2(B+\Delta)^2\exp(-k \tfrac{\delta^2}{(B+\delta)(B+\delta+1)})
\end{align}
where we used that for $R,s>0$ the estimate $D(R+s\|R)\geq\tfrac{1}{2}\tfrac{s^2}{(R+s)(R+s+1)}$ must hold. In addition, by assumption, $(u_k-M)^2+(v_k-N)^2\geq2\delta^2$.


Thus setting $\delta=s\Delta$ and utilizing $\Delta<B/2$, the probability of the event $j\in\mathbf A$ is lower bounded as
\begin{align}
    \mathbb P(\mathbf A\cap \mathbf B)&\geq1-\exp(-k \tfrac{\Delta^2s^2}{(B(1+s/2))(B(1+s/2)+1+s/2)}) \\
    &\geq1-\exp(-k\tfrac{s^2}{(1+s/2)^2}\tfrac{\Delta^2}{B(B+1)}+\mathcal O(\log k))
\end{align}
Further, we have 
\begin{align}
    \mathbb P(\mathbf B^\complement)
        &\leq \exp\left(-k\tfrac{\Delta^2}{7B(B+1)}\right).
\end{align}
Given however that $j\in\mathbf A$ we have, following \eqref{eqn:nu_geq_eps|j,t} and using $\sin(x)>x(1-x)$ for small enough non-negative $x$, the estimate  
\begin{align}
    \mathbb P(\nu\geq\epsilon|j,t) 
        &\leq \exp(k\cdot2(\Delta/2-s\Delta)\log[1-(1-1/q)\sin^2(\epsilon/2)])\\
        &\leq \exp(k\cdot(\Delta(1-2s))\log[1-\tfrac{M-N}{M(N+1)}(1-\epsilon)^2\epsilon^2/4])\\
        &\leq \exp(-k\cdot(\Delta\tfrac{(1-2s)}{4})\tfrac{\Delta}{B(B+1)}(1-\epsilon)^2\epsilon^2])\\
        &= \exp(-k\cdot\tfrac{(1-2s)}{4}\tfrac{\Delta^2}{B(B+1)}(1-\epsilon)^2\epsilon^2]).
\end{align}
By picking $s=1/3$ and using $(1-\epsilon)^2\epsilon^2\geq(7/8)\epsilon^2$ for $\epsilon$ small enough, we obtain that the probability $\mathbb P(\nu\geq\epsilon) $ can, including an abort if $t$ is too large or $j$ too small, be bounded as 
\begin{align}
    \mathbb P(\nu\geq\epsilon|(j,t)\in\mathbf A\times\mathbf B) \leq \exp(-k\tfrac{\Delta^2\epsilon^2}{16B(B+1)}+\mathcal O(\log k)).
\end{align}
At the same time, for the abort conditions it holds
\begin{align}
    \max\{\mathbb P(\mathbf A^\complement),\mathbb P(\mathbf B^\complement)\} \leq \exp(-k\tfrac{\Delta^2\epsilon^2}{16B(B+1)}+\mathcal O(\log k)).
\end{align}
As can be seen, better constants in the order of $1/14$ instead of $1/16$ are possible as well.

\paragraph{\gls{tdma} Photon Counting}

We define the bounded random variable,  $r$-th repetition of setting $i$, 
\begin{equation}
    Y_{i,r} := \begin{cases}
        +1, & \text{if the outcome is }(1,0),\\
        -1, & \text{if the outcome is }(0,1),\\
        0,  & \text{otherwise}.
    \end{cases}
    \label{eq:tdma_random_variable}
\end{equation}
 Then
\begin{equation}
    \mathbb{E}[Y_{i,r}] = p_i(1,0)-p_i(0,1) = \frac{M-N}{(M+1)(N+1)}n_i,\label{eq:tdma_mean}
\end{equation}
\begin{equation}
    \mathbb{E}[Y_{i,r}^{2}]
    =
    p_i(1,0)+p_i(0,1)
    =
    \frac{M+N+2MN}{(M+1)(N+1)}.
    \label{eq:tdma_second_moment}
\end{equation}
Therefore, each setting gives one component of $\boldsymbol{n}$ with some factor. 
For notational simplicity, we define:
\begin{equation}
    \lambda := \frac{M-N}{(M+1)(N+1)}, \qquad      q:= \frac{M+N+2MN}{(M+1)(N+1)}
\end{equation}
Hence,
\begin{equation}
  \operatorname{Var}(Y_{i,r})=q-\lambda^2n_i^2\leq q  , \qquad |Y_{i,r}-\mathbb{E}Y_{i,r}|\leq1+\lambda\leq2, \qquad q-\lambda=\frac{2N(1+M)}{(M+1)(N+1)}\geq0
\end{equation}
which are used below to control the Bernstein denominator.

\paragraph{Estimator.}
The estimator is as follows:
    \begin{equation}
    \widehat{s}_i :=\frac{1}{m}\sum_{r=1}^{m}Y_{i,r},
\end{equation}
and, we collect $\widehat{\boldsymbol{s}}=(\widehat{s}_x,\widehat{s}_y,\widehat{s}_z)$ such that, 
\begin{equation}
    \mathbb{E}[\widehat{\boldsymbol{s}}]=\lambda\boldsymbol{n}. 
\end{equation}
Therefore, the estimator on average does point towards the correct direction, but with a scale, which can be removed by normalization:
\begin{equation}
    \widehat{\boldsymbol{n}}_{\mathrm{pc}}
    :=
    \frac{\widehat{\boldsymbol{s}}}{\|\widehat{\boldsymbol{s}}\|_2},
    \label{eq:tdma_normalized_estimator}
\end{equation}
From the standard triangle inequality, we get:
\begin{equation}
    \bigl\|\widehat{\boldsymbol{n}}_{\mathrm{pc}}-\boldsymbol{n}\bigr\|_2
    \leq
    \frac{2}{\lambda}
    \bigl\|\widehat{\boldsymbol{s}}-\lambda\boldsymbol{n}\bigr\|_2 .
    \label{eq:tdma_normalization_bound}
\end{equation}
A vector in $\mathbb{R}^3$ of norm at least $\lambda\epsilon/2$ has a coordinate of modulus at least $\lambda\epsilon/(2\sqrt{3})$, hence
\begin{equation}
    \bigl\{
        \|\widehat{\boldsymbol{n}}_{\mathrm{pc}}-\boldsymbol{n}\|_2\geq\epsilon
    \bigr\}
    \subseteq
    \bigcup_{i\in\{x,y,z\}}
    \left\{
        \bigl|\widehat{s}_i-\lambda n_i\bigr|
        \geq
        \frac{\lambda\epsilon}{2\sqrt{3}}
    \right\}.
    \label{eq:tdma_event_inclusion}
\end{equation}

Since $ |Y_{i,r}-\mathbb{E}Y_{i,r}|
    \leq2$, and \(\operatorname{Var}(Y_{i,r})\leq q\),  Bernstein's inequality (see e.g. \cite{BoucheronLugosiMassart2013}) gives
\begin{equation}
    \Pr
    \left(
        \left|
            \widehat{s}_i-\lambda n_i
        \right|
        \geq t
    \right)
    \leq
    2\exp
    \left[
        -
        \frac{
            mt^2
        }{
            2\left(
                q+\frac{2t}{3}
            \right)
        }
    \right].
    \label{eq:tdma_bernstein}
\end{equation}
Setting $t=\lambda\epsilon/(2\sqrt{3})$, a union bound over the three settings
and $m=k/3$ yield the nonasymptotic bound
\begin{align}
    \mathbb P\left(
        \bigl\|\widehat{\boldsymbol{n}}_{\mathrm{pc}}-\boldsymbol{n}\bigr\|_2
        \geq\epsilon
    \right) &\leq6\exp(-\frac{k\lambda^2\epsilon^2}{144\left(q+\frac{\lambda\epsilon}{3\sqrt{3}}\right)})\label{eq:tdma_finite_sample_bound}\\
    &\leq 6\exp(-\frac{k\Delta^2\epsilon^2}{144B(B+1)^3})
\end{align}
which applies whenever 
\paragraph{Heterodyne- and Homodyne Unit Cell}

The covariance matrix of the thermal state $S_M\otimes S_N$ is given by 
\begin{align}
    \sigma_{M,N} = \left(\begin{array}{cccc} 
                            2M+1 & 0 & 0 & 0\\
                            0 & 2M+1 & 0 & 0 \\
                            0 & 0 & 2N+1 & 0\\
                            0 & 0 & 0 & 2N+1
                        \end{array}\right)
                = \left(\begin{array}{cc} 
                            a\eins & 0 \\
                            0 & b\eins 
                        \end{array}\right)
\end{align}
with $a=2M+1$ and $b=2N+1$. According to \cite[Eq. (5.9)]{bookserfini} this matrix transforms as
\begin{align}
    \sigma_{M,N}\to S\sigma_{M,N}S^T
\end{align}
under a quadratic unitary. For the phase-shifter on mode $A$, this transformation is described by \cite[Eq. (5.12)]{bookserfini}
\begin{align}
    S_\phi =\left(\begin{array}{cccc} 
                            \cos(\phi) & \sin(\phi) & 0 & 0\\
                            -\sin(\phi) & \cos(\phi) & 0 & 0 \\
                            0 & 0 & 1 & 0\\
                            0 & 0 & 0 & 1\\
            \end{array}\right)
            = \left(\begin{array}{cc} 
                            V(\phi) & 0 \\
                            0 & \eins 
                        \end{array}\right)
\end{align}
and for the beam-splitter we have by \cite[Eq. (5.13)]{bookserfini}
\begin{align}
    S_\theta =\left(\begin{array}{cccc} 
                            \cos(\theta) & 0 & \sin(\theta) & 0\\
                            0 & \cos(\theta) & 0 & \sin(\theta)\\
                            -\sin(\theta) & 0 & \cos(\theta) & 0\\
                            0 & -\sin(\theta) & 0 & \cos(\theta)\\
            \end{array}\right)
            =\left(\begin{array}{cc} 
                            \cos(\theta)\eins & \sin(\theta)\eins\\
                            -\sin(\theta)\eins & \cos(\theta)\eins
            \end{array}\right).
\end{align}
Setting $\sigma_{M,N,\theta}:=S_\theta\sigma_{M,N}S_\theta^T$ this yields 
\begin{align}
    S_\theta \sigma_{M,N} S_\theta^{T}
    &=
    \left(\begin{array}{cc}
            (c^2a+s^2b)\eins&sc(b-a)\eins\\
            sc(b-a)\eins&(s^2a+c^2b)\eins
    \end{array}\right),
\end{align}
which then transforms to 
\begin{align}
    \sigma_{M,N,\theta,\phi}:=S_\phi S_\theta \sigma_{M,N} S_\theta^{T}S_\phi^T
    &=
    \left(\begin{array}{cc}
            (c^2a+s^2b)\eins&sc(b-a)V(\phi)\\
            sc(b-a)V(\phi)^T&(s^2a+c^2b)\eins
    \end{array}\right)
\end{align}
under the action of the phase shifter. The covariance matrix of the heterodyne measurement yields a mean-zero multimode normal distribution with covariance matrix 
\begin{align}
    \cov_\mathrm{het}(M,N,\theta,\phi)=\tfrac{1}{2}(\sigma_{M,N,\theta,\phi}+\eins)
\end{align}
We define 
\begin{align}
    A_x &= \frac{1}{2}
        \begin{pmatrix}
            0  & 0  & -1 & 0  \\
            0  & 0  & 0  & -1 \\
            -1 & 0  & 0  & 0  \\
            0  & -1 & 0  & 0
        \end{pmatrix},
    \\[1ex]
    A_y &= \frac{1}{2}
            \begin{pmatrix}
                0  & 0 & 0 & -1 \\
                0  & 0 & 1 & 0  \\
                0  & 1 & 0 & 0  \\
                -1 & 0 & 0 & 0
            \end{pmatrix},
    \\[1ex]
    A_z &= \frac{1}{2}
            \begin{pmatrix}
                1 & 0 & 0  & 0  \\
                0 & 1 & 0  & 0  \\
                0 & 0 & -1 & 0  \\
                0 & 0 & 0  & -1
            \end{pmatrix}.
\end{align}
Writing $y^4=(x_1,p_1,x_2,p_2)^T$ we can then estimate the direction of $\mathbf n$ from a series $y^k=(y_1^4,\ldots,y_k^4)$ via 
\begin{align}
    \hat d_i = k^{-1}\sum_{j=1}^k\langle y^4_j,A_iy^4_j\rangle
\end{align}
and the full vector $\mathbf n$ as 
$\hat{\mathbf n}_i=\|\hat d\|_2^{-1}\hat d_i$. The expectation values for the estimators $d_x,d_y,d_z$ are explicitly given by 
\begin{align}
    d_i = (M-N)n_i,\qquad i=x,y,z, 
\end{align}
as a consequence of the formula $\mathbb E(G_iG_j)=\cov(G_iG_j)$ which holds for the expectation values of a correlated mean-zero normal distribution. We apply the Hanson-Wright Theorem \cite{rudelson2013hanson} after mapping $A_i$ to $\sqrt{\cov}A_i\sqrt{\cov}$ and, accordingly, $y^4$ to $y^{'4}=\cov^{-1/2}y^4$, so that the components of $y^{'4}$ are independent. It then holds
\begin{align}
    \mathbb P(|d_i - (M-N)n_i|\geq \epsilon)&\leq 2\exp(-k \tfrac{\epsilon^2}{\|\sqrt{\cov}A_i\sqrt{\cov}\|_F^2})
\end{align}
whenever $B+1 > \Delta\epsilon/\sqrt{12}$. Since it holds $\|\sqrt{\cov}A_i\sqrt{\cov}\|_F^2=(M+1)(N+1)+(M-N)n_i^2/2\geq(M+1)(N+1)$. Thus a favorable universal performance of the heterodyne-based estimator when normalizing by $(M-N)^{-1}$ is, by application of a union bound, given by
\begin{align}
    \mathbb P(\|\omega - \mathbf n\|\geq \epsilon)&\leq 6\exp(-k\cdot \tfrac{\epsilon^2(M-N)^2/3}{(B+1)^2})\\
        &\leq 6\exp(-k\tfrac{\epsilon^2\Delta^2}{3(B+1)^2}).
\end{align}
The analysis of the homodyne exponent follows a similar reasoning, but the $2$-mode Gaussian distribution has a slightly lower shot noise contribution of $1/2$ instead of $1$.
\paragraph{Estimation Error}
    The result follows from bounding the residual from above by $\tfrac{M-N}{2}\sqrt{1-\langle\mathbf n,\omega\rangle^2}$, which by the Cauchy-Schwarz inequality and the bound $M-N\leq B$. 

\end{proof}

\section{Conclusion and Outlook}

Using the framework of quantum multiparameter estimation, we have addressed the problem of separating thermal optical sources mixed by an unknown passive linear transformation. As a conventional baseline, we considered heterodyne detection, in which the covariance matrix is first estimated and subsequently diagonalized on a classical computer. In contrast, our method estimates and suppresses correlations directly within the optical device, thereby implementing an in situ diagonalization of the optical correlation matrix.

We find that quantum measurement theory can significantly improve the learning of a passive optical transformation even when the sources themselves are classical thermal states. Assuming the ability to implement Holevo-optimal local sensing, combining it with a self-configuring Jacobi network within our framework yields a practical quantum-optical unmixing transformation. Comparing with the heterodyne baseline, we identify the weak-light regime in which this approach offers its strongest advantage.

Several questions remain open. A central practical challenge is the realization of the required collective photon-number and spin measurements using experimentally available optical components and realistic detectors. From an algorithmic perspective, adaptive schedules and parallel architectures may reduce the total number of copies required \cite{11479855}. Further, the stabilization of the required optical architectures poses fundamental engineering challenges \cite{Litvin:25,Litvin:26}. Finally, extending the framework to nonthermal, non-Gaussian, or fermionic inputs \cite{Walschaers_2021, Adhikari_2024} and quantum communication networks \cite{10461354} may broaden the class of source-separation problems for which a quantum advantage can be achieved.

\section*{Acknowledgements}
The research is part of the Munich Quantum Valley, which is supported by the Bavarian state government with funds from the Hightech Agenda Bayern Plus. This work was financed by the DFG via grant NO 1129/2-1 and by the Federal Ministry of Education and Research of Germany in the Q-STARS project, grant number 16KIS2604, as well as via grants 16KISQ093, 16KISQ039 and 16KISQ077. The generous support of the state of Bavaria via the 6GQT project is greatly appreciated.

During preparation of this manuscript, the authors used ChatGPT (OpenAI) to assist with mathematical cross-checking, literature discovery, and language editing. All mathematical statements, derivations, references, and conclusions were independently verified by the authors, who take full responsibility for the content of the manuscript.

\bibliographystyle{plain}
\bibliography{bib}

\appendix

\section{Technical Proofs}
\label{appendix:technical_proofs}
\subsection{Eigenvalues}
The eigen-projections $E_j$ are derived as follows:
    Fix photon numbers $s$ and $t$. Then for every two photon number strings $\ba=(a_1,\ldots,a_k)$ and $\bb=(b_1,\ldots,b_k)$ with $|\ba|=s$ and $|\bb|=t$ we have
    \begin{align}
        \mathcal{J}|\ba,\bb\rangle &= \left[J_z^2 + \tfrac{1}{2}(J_+J_-+J_-J_+)\right]|\ba,\bb\rangle \\
            &= \tfrac{1}{4}(s-t)^2|\ba.\bb\rangle + \tfrac{1}{2}\big[J_+\sum_i\sqrt{(a_i+1)b_i}|\ba+e_i,\bb-e_i\rangle + \nonumber\\
            &\qquad + J_-\sum_i\sqrt{(b_i+1)a_i}|\ba-e_i,\bb+e_i\rangle)\big]\\
            &= \tfrac{1}{4}(s-t)^2|\ba,\bb\rangle + \\
            &\ \ \ + \tfrac{1}{2}\big[\sum_{i,j}\sqrt{(a_i+1)b_i(a_j+\delta_{ij})(b_j-\delta_{ij}+1)}|\ba+e_i-e_j,\bb-e_i+e_j\rangle + \nonumber\\ &\ \ \ + \sum_{i,j}\sqrt{(b_i+1)a_i(a_j-\delta_{ij}+1)(b_j+\delta_{ij})}|\ba-e_i+e_j,\bb+e_i-e_j\rangle)\big],\nonumber
    \end{align}
    where $e_i=(0,\ldots,0,1,0,\ldots,0)$ has the entry $1$ in the $i$-th position. As one can see, $\mathcal{J}$ preserves the subspaces that the projectors $P_s\otimes P_t$ project onto. The operators $J_x,J_y,J_z$ fulfill the commutation relations of algebra of angular momentum. Hence their joint eigenvectors are labeled as $(j,j_z)$ and their joint eigenstates $|j,j_z\rangle$ fulfill \cite{ZwiebachAngularMomentum2013}
    \begin{align}
        J^2|j,j_z\rangle &= j(j+1)|j,j_z\rangle,\qquad J_z|j,j_z\rangle=j_z|j,j_z\rangle.
    \end{align}
    Further we have, again following the exposition in \cite{ZwiebachAngularMomentum2013},
    \begin{align}
        J_\pm|j,j_z\rangle = \sqrt{j(j+1)-j_z(j_z\pm1)}|j,j_z\pm1\rangle
    \end{align}
    where $j\in\{(i-1)/2\}_{i\in\mathbb N}$ and $-j\leq j_z\leq j$. 
    To obtain the eigenvalues of $\mathcal{J}$, we write 
    \begin{align}
        \mathcal{J} |j,j_z\rangle &=(J_z^2+J_+J_--J_z)|j,j_z\rangle\label{eqn:different-version-of-C}\\
            &= \left(j_z(j_z-1) + j(j+1)-j_z(j_z-1)\right)|j,j_z\rangle\\
            &= j(j+1)|j,j_z\rangle.
    \end{align}   
    Thus, $\mathcal{J}=J(J+1)$.
    
\subsection{Measurement}

    \begin{proof}[Proof of Lemma \ref{lem:alternative-representations}]
        We prove \eqref{eqn:hatM} by using $D_z:=J - J_z$, leading to 
        \begin{align}
            k\cdot \hat M &= \tfrac{1}{2}N_{AB}+J \\
                &= \tfrac{1}{2}N_{AB}+J_z +D_z\\
                &= N_{A}+D_z\\
        \end{align}
        For \eqref{eqn:hatTau} we first define $J_z^o:=U(\theta)^{\otimes k}\cdot J_z\cdot U(\theta)^{\dagger\otimes k}$, which can using $c:=\cos(2\theta)$ and $s:=\sin(2\theta)$ be written as 
        \begin{align}
            J_z^o = c\cdot J_z + s\cdot J_x.
        \end{align}
        Since $[J,U(\theta)^{\otimes k}]=0$ and $\tau=(1+c)/2$ it follows
        \begin{align}
            \hat\tau &= \tfrac{1}{2}(1+U(\theta)^{\dagger\otimes k}J_z^oJ^{-1}U(\theta)^{\otimes k})\\
                &= \tfrac{1}{2}(1+U(\theta)^{\dagger\otimes k}\big(c\cdot J_z + s\cdot J_x\big)J^{-1}U(\theta)^{\otimes k})\\
                &= \tau + U(\theta)^{\dagger\otimes k}\big(\tfrac{c}{2}J_z + \tfrac{s}{2}J_x-\tfrac{c}{2}J\big)J^{-1}U(\theta)^{\otimes k})\\
                &= \tau + U(\theta)^{\dagger\otimes k}\big(\tfrac{s}{2}J_x-\tfrac{c}{2}D_z\big)J^{-1}U(\theta)^{\otimes k})
        \end{align}
        
    \end{proof}

        \begin{proof}[Proof of Lemma \ref{lem:bound-onJx2J-2}]
            We use the abbreviations $x=M/(M+1)$, $y=N/(N+1)$, $q=x/y$ $c=(1-x)^k(1-y)^k$. Since $M>N$, it then holds $q>1$.
            \paragraph{Proof of Equation \eqref{eqn:bound-onJx2J-2}}
            \begin{align}
                \Tr(J_x^2J^{-2}\rho_0) 
                    &=c\sum_{t,j}\sum_{m=-j}^jx^{\tfrac{t}{2}+m}y^{\tfrac{t}{2}-m}\Tr(J_x^2J^{-2}[P_{\tfrac{t}{2}+m}\otimes P_{\tfrac{t}{2}-m}]E_j)\\
                    &=c\sum_{t,j}\sum_{m=-j}^jx^{\tfrac{t}{2}+m}y^{\tfrac{t}{2}-m}\Tr(J_x^2[P_{\tfrac{t}{2}+m}\otimes P_{\tfrac{t}{2}-m}]E_j)\cdot \tfrac{1}{j^2}.
            \end{align}
            On the subspace $V_{t,m,j}$ corresponding to the projection $[P_{t/2+m}\otimes P_{t/2-m}]E_j$ the operators $\mathcal{J}$ and $J_z$ are constant, hence $V_{t,m,j}$ is a spanned by vectors $|j,m\rangle$. For each such vector, we have
            \begin{align}
                J_x|j,m\rangle &= \frac{1}{2}\sqrt{(j-m)(j+m+1)}|j,m+1\rangle + \tfrac{1}{2}\sqrt{(j+m)(j-m+1)}|j,m-1\rangle.
            \end{align}
            In particular for $J_x^2$ and $J_x^4$ this implies 
            \begin{align}
                \langle j,m|J_x^2|j,m\rangle &= \tfrac{1}{2}\left(j(j+1) - m^2\right)\\
                \langle j,m|J_x^4|j,m\rangle &= \tfrac{1}{8}\big(3(j(j+1)-m^2)^2-2j(j+1)+5 m^2\big).
            \end{align}
            By \eqref{eqn:combinatorial-dim(S_lambda)} $\Tr([P_{t/2+m}\otimes P_{t/2-m}]E_j)=\mu_{j,t}$ is independent of $m$.  and therefore  
            \begin{align}
                \Tr(J_x^2[P_{t/2+m}\otimes P_{t/2-m}]E_j)
                    &= \mu_{j,t}\frac{j(j+1)-m^2}{2}
           \end{align}
           We proceed with an estimate on the the sum over $d$: Let $S:=\sum_{d=0}^{2j}q^{-d}$ and set $p_d:=q^{-d}/S$. Then
           \begin{align}
                   S^{-1}\sum_{d=0}^{2j}q^{-d}\frac{j+2dj-d^2}{2j^2}&=\sum_{d=0}^{2j}p_d\frac{j+2dj-d^2}{2j^2}\\
                    &\leq \tfrac{1}{2j} + \mathbb E(d)\cdot\tfrac{1}{j}\\
                    &\leq \tfrac{1}{2j} + \tfrac{1}{j(q-1)}\\
                    &=\tfrac{1}{j}\tfrac{q+1}{2(q-1)}
               \end{align}
            which holds for every $j$ and is based on the convergence of the arithmetico-geometric sequence $\sum_{d=0}^\infty dr^k=r/(1-r)^2$ for values $r\in(-1,1)$. Thereby, we can transform our estimate of $\Tr(J_x^2J^{-2}\rho_0)$ to       
            \begin{align}
                \Tr(J_x^2J^{-2}\rho_0) &\leq\frac{1+q}{2(q-1)}c\sum_{t,j}(xy)^{t/2}\mu_{t,j}q^{j}\sum_{d=0}^{2j}q^{-d}\tfrac{1}{j}\\
                    &\leq\frac{q+1}{2(q-1)}\Tr(J^{-1}\rho_0)
           \end{align}
           
           \paragraph{Proof of Equation \eqref{eqn:bound-onJx4J-4}}
           In order to bound $\Tr(J_x^4J^{-4}\rho_0)$ we first define
           \begin{align}
               \mathrm{pl}(d,j):=\tfrac{12j^2d^2+12j^2d+6j^2-12jd^3-6jd^2-10jd-2j+3d^4+5d^2}{j^2}
           \end{align}
           and then derive 
            \begin{align}
                \Tr(J_x^4J^{-4}\rho_0)&=\sum_{t,j}\sum_{m=-j}^j\Tr(J_x^4J^{-4}[P_{t/2+m}\otimes P_{t/2-m}]E_j\rho_0)\nonumber\\
                    &= c\sum_{t,j}(xy)^{t/2}\mu_{j,t}\frac{1}{8j^4}\sum_{m=-j}^jq^m\big(3(j(j+1)-m^2)^2-2j(j+1)+5 m^2\big)\\
                    &\ \ = c\sum_{t,j}(xy)^{t/2}\mu_{j,t}q^{j}\frac{1}{8j^2}\sum_{d=0}^{2j}q^{-d}\mathrm{pl}(d,j)\\
                    &\ \ \leq c\sum_{t,j}(xy)^{t/2}\mu_{j,t}q^{j}\frac{1}{8j^2}\sum_{d=0}^{2j}q^{-d}\frac{12j^2d^2+12j^2d+6j^2}{j^2}\\
                    &\ \ \leq c\sum_{t,j}(xy)^{t/2}\mu_{j,t}q^{j}\frac{1}{8j^2}\sum_{d=0}^{2j}q^{-d}3((2d+1)^2+1)
                    &\ \ =\frac{3}{4}(\frac{q+1}{q-1})^2\Tr(J^{-2}\rho_0)
           \end{align}
           where we used $m=j-d$, implying $j(j+1)-(j-d)^2= j(2d+1)-d^2$, $d\leq 2j$ and again the  arithmetico-geometric series.

           \paragraph{Proof of Estimate \eqref{eqn:bound-onD_z2J-2}}
           Since $D_z=J-J_z$ commutes with $J$, we have with $\rho_0=S_M^{\otimes k}\otimes S_N^{\otimes k}$
           \begin{align}
               &\Tr(D_z^tJ^{-t}\rho_0) = c\sum_{n,j}\sum_{m=-j}^jx^{\tfrac{n}{2}+m}y^{\tfrac{n}{2}-m}\Tr(D_z^tJ^{-t}[P_{\tfrac{n}{2}+m}\otimes P_{\tfrac{n}{2}-m}]E_j)\\
                &=c\sum_{n,j}\sum_{m=-j}^jx^{n/2+m}y^{n/2-m}\frac{(j-m)^t}{j^t}\Tr([P_{n/2+m}\otimes P_{n/2-m}]E_j)\\
                &=c\sum_{n,j}\frac{q^j}{j^t}x^{n/2}y^{n/2}\Big(\sum_{d=0}^{2j}d^tq^{-d}\Big)\mu_{j,n}\\
                &\leq C_{M,N,t}\cdot c\sum_{n,j}\frac{q^j}{j^t}x^{n/2}y^{n/2}\Tr([P_{n/2+j}\otimes P_{n/2-j}]E_j)\\
                &\leq C_{M,N,t}\Tr(J^{-t}\rho_0)\\
                &\leq C_{M,N,t}\left(\Big(\frac{2}{k(M-N-\epsilon)}\Big)^t  + o(k^{-t})\right)
           \end{align}
           where we used Lemma \ref{lem:bound-for-Jinverse} in the last step and $C_{M,N,t}=L_{-t}(1/q)$ as well as the dimension formula \eqref{eqn:combinatorial-dim(S_lambda)} where $L_{-t}(1/q)$ is the polylogarithm \cite{polylogarithm}. 
        \end{proof}
        \begin{proof}[Proof of Lemma \ref{lem:bound-for-Jinverse}]
            Let $\eins_{m\neq0}$ be the subspace on which $J_z\neq0$. Then by the eigenvalue estimate $-j\leq j_z\leq j$ it holds $J\eins_{m\neq0}\geq|J_z|\eins_{m\neq0}$. The space on which the eigenvalues of $J_z$ equal zero is exactly the direct sum of subspaces $V_s\otimes V_s$ of equal photon numbers, where $\psi\in V_s\otimes V_s$ inevitably implies $(N_A-N_B)\psi=0$. Write $J=J^\parallel+J^\perp$, where $J^\parallel$ has support only on the orthocomplement of $\mathrm{ker}(J_z^2)$ and $J^\perp$ only on the $\mathrm{ker}(J_z^2)$. Then using generalized inverses, we have $J^{-1}=J^{\parallel-1}+J^{\perp-1}$. It follows 
            \begin{align}
                \Tr(J^{\parallel-t}\rho_0) 
                    &\leq \Tr(|J_z|^{-t}\eins_{m\neq0}\rho_0)\\
                    &=2^t\cdot\mathbb E_{A^k B^k}\left((1-\delta_{N(A^k)=N(B^k)})\cdot |N(A^k)-N(B^k)|^{-t}\right)
            \end{align}
            where the expectation is with respect to the thermal product measures $p_M^{\otimes k}$ and $p_N^{\otimes k}$, with $p_X(n):=\frac{1}{X+1}(\tfrac{X}{X+1})^n$ and $N(A^k)$ ($N(B^k)$) denotes the photon number in the length-$k$ string $A^k$ ($N(B^k)$). The functions $N(x^k):=\sum_ix_i$ satisfy, upon using the Chernoff bound,
            \begin{align}
                \mathbb P(X)\leq e^{-k\cdot C'_{M,N}}
            \end{align}
            where $X=\{| N(A^k) - N(B^k)- k(M-N)|\geq k\epsilon\}$ and $C'_{M,N,\epsilon}>0$ depends on $M$ and $N$ as well as $\epsilon$. Let $\Xi_{\neq0}$ be the set of $(a^k,b^k)$ such that $N(a^k)\neq N(b^k)$. On this set, $|N(a^k)-N(b^k)|\geq1$. Then
            \begin{align}
                \Tr(J^{\parallel-t}\rho_0) &\leq2^t\cdot\mathbb E_{A^k B^k}\left((1-\delta_{N(A^k)=N(B^k)})\cdot |N(A^k)-N(B^k)|^{-t}\right)\\
                &\leq \frac{2^t}{k^t(M-N-\epsilon)^t} + 2^t\mathbb E_{X\cap\Xi_{\neq0}}|N(A^k)-N(B^k)|^{-t}\\
                    &\leq \frac{2^t}{k^t(M-N-\epsilon)^t} + 2^t\mathbb P(X\cap\Xi_{\neq0})\\
                    &\leq \frac{2^t}{k^t(M-N-\epsilon)^t} + e^{-k\cdot C'_{M,N,\epsilon}}.
            \end{align}
            On the space where $J_z^2=0$ the smallest non-zero eigenvalue of $J$ is $1/2$, hence $J^{\perp-t}\leq2^{t}\eins_{\mathrm{ker}(J_z^2)}$. However by using the Chernoff bound, we obtain
            \begin{align}
                \Tr(\eins_{\mathrm{ker}(J_z^2)}\rho_0) &= \sum_{N(A^k)=N(B^k)}p_M^{\otimes k}(A^k)p_N^{\otimes k}(B^k)\\
                    &\leq\exp(-k\tfrac{\Delta^2}{5B(B+1)}).
            \end{align}
            Thus $\Tr(J^{\perp-t}\rho_0)\leq\exp(-k\tfrac{\Delta^2}{5B(B+1)})$.
        \end{proof}
        \begin{proof}[Proof of Lemma \ref{lem:bound-on-Dz}]
            We use the abbreviations $x=M/(M+1)$, $y=N/(N+1)$, $q=x/y$ $c=(1-x)^k(1-y)^k$. Since $M>N$, it then holds $q>1$.
            Recall that $D_z = J - J_z$, $J_z=(N_A-N_B)/2$ and let $E_j$ be the spectral projections of $J$. Thus if $t\geq1$, 
            \begin{align}
                &\Tr(D_z^tS_M^{\otimes k}\otimes S_N^{\otimes k})
                    =c\sum_{u,v}x^uy^v\Tr(D_z^t[P_u\otimes P_v])\\
                    &=c\sum_{n,j}\sum_{m=-j}^jx^{n/2+m}y^{n/2-m}\Tr(D_z^t[P_{n/2+m}\otimes P_{n/2-m}]E_j)\\
                    &=c\sum_{n,j}\sum_{m=-j}^jx^{n/2+m}y^{n/2-m}(j-m)^t\Tr([P_{n/2+m}\otimes P_{n/2-m}]E_j)\\
                    &=c\sum_{n,j}\sum_{d=0}^{2j}x^{n/2+j}y^{n/2-j}q^{-d}d^t\mu_{n,j}\\
                    &\leq c\sum_{n,j}x^{n/2+j}y^{n/2-j}\mu_{n,j}\sum_{d=0}^\infty q^{-d}d^t\\
                    &=c\sum_{n,j}x^{n/2+j}y^{n/2-j}\mu_{n,j}L_{-t}(1/q)\\
                    &\leq L_{-t}(1/q)
           \end{align}
           where $L_{-t}(1/q)$ is the polylogarithm \cite{polylogarithm}. The derivative equals 
           \begin{align}
               \partial_M\Tr(D_z\rho_0)&=\frac{1}{M(M+1)}\Tr(D_z(N_A-kM)\rho_0)\\
                        &= \frac{1}{M(M+1)}\cov(D_z,N_A)\\
                        &\leq\mathcal O(\sqrt{k})
           \end{align}  
           by above bounds, and the same holds for $\partial_M\Tr(D_z\rho_0)$.        
        \end{proof}

    \subsection{Holevo Optimality}\label{app:holevo-optimality}

        \begin{proof}[Proof of Lemma \ref{lem:first-and-second-order-spin-momemnts}]
            We conider the case of a single irreducible representation first, where the equality $J_\omega|j,\omega\rangle=j|j,\omega\rangle$ uniquely defines $|j,\omega\rangle$. In this case, it is possible to define $M_j(\omega):=\frac{2j+1}{4\pi}|j,\omega\rangle\langle j,\omega|$. From \cite[Eqns. (1)]{Giraud_2008} we then have the following:
            \begin{align}
                \int M_j(\omega) &= E_j
            \end{align}
            From \cite[Eq. (A5)]{Manai_2023} we have 
            \begin{align}
                \langle j,\omega|J|j,\omega\rangle = j\cdot\omega .
            \end{align}
            while \cite[Eq. (A6)]{Manai_2023} gives
            \begin{equation}
                \mathbf J = \frac{2j+1}{4\pi}\int_{S^2}(j+1)\cdot\omega\cdot |j,\omega\rangle\langle j,\omega|d\omega .
            \end{equation}
            and finally \cite[Eq. (A7)]{Manai_2023} translates to 
            \begin{align}
                J_z^2 = \frac{2j+1}{4\pi}\int_{S^2}\left[(j+1)\left(j+\frac{3}{2}\right)\omega_z^2-\frac{j+1}{2}\right]|j,\omega\rangle\langle j,\omega|d\omega .
            \end{align}
            which finally gives by rotational invariance and via $M_j(d\omega)=\tfrac{2j+1}{4\pi}|j,\omega\rangle\langle j,\omega|d\omega$ the equation 
            \begin{align}
                \int_{S^2}\langle\mathbf n,\omega\rangle M_j(d\omega)
                    &= \frac{J_\mathbf{n}}{j+1}\\
                \int_{S^2}\langle\mathbf n,\omega\rangle^2\,M_j(d\omega)
                    &=\frac{2J_\mathbf{n}^2}{(j+1)(2j+3)} + \frac{1}{2j+3}\,\mathbb I_j
            \end{align}
            for every unit vector $\mathbf n$. Vectors of arbitrary length are handled by rescaling of the formula.
            To obtain the cross-terms, we write $Q(\mathbf n):=\int_{S^2}\langle \mathbf n,\omega\rangle^2\,M_j(d\omega)$ and then for $i\neq j$
            \begin{align}
                2\cdot Q(\tfrac{1}{\sqrt{2}}(e_i+e_l)) = Q(e_i)+Q(e_l) + 2\int_{S^2}\omega_i\omega_l\cdot M_j(d\omega),
            \end{align}
            so that the form as used by us can be completely derived from \cite{Manai_2023}. Alternatively, the formulas can be derived by explicit calculation using \cite[Appendix D]{arcchiAtomicCoherentStatesInQuantumOptics}, in particular equations D16 and D17, as well as D6.  

            Since all spin-$j$ representations have exactly the same dimension $\mu_{j,t}$, all arguments carry over to the statement of Lemma \ref{lem:first-and-second-order-spin-momemnts}
        \end{proof}

        \begin{proof}[Proof of Lemma \ref{lem:abba}]
            Let
            \begin{align}
            A:=J_xJ^{-1},\qquad B:=D_zJ^{-1},\qquad D_z:=J-J_z .
            \end{align}
            We work on a fixed spin-\(j\) sector, \(j>0\). Then \(J^{-1}=j^{-1}\) on this sector, and
            \begin{align}
            B|j,m\rangle=\frac{j-m}{j}|j,m\rangle .
            \end{align}
            We set $d=j-m$, so that $m=j-d$. With this relabelling of parameters, we get 
            \begin{align}
            B|j,m\rangle=\frac{d}{j}|j,m\rangle .
            \end{align}
            It holds $J_x=\frac{J_++J_-}{2}$, therefore
            \begin{align}
                J_+|j,m\rangle &=\sqrt{(j-m)(j+m+1)}\,|j,m+1\rangle\\
                J_-|j,m\rangle &=\sqrt{(j+m)(j-m+1)}\,|j,m-1\rangle.
            \end{align}
            Hence
            \begin{align}
                A|j,m\rangle &= \frac{1}{2j}\Big[ \sqrt{(j-m)(j+m+1)}|j,m+1\rangle + \nonumber\\
                &\qquad \qquad + \sqrt{(j+m)(j-m+1)}|j,m-1\rangle \Big].
            \end{align}
            We proceed to computing $\langle j,m|ABAB|j,m\rangle$. From $B|j,m\rangle = \frac{j-m}{j}|j,m\rangle$ we derive 
            \begin{align}
                AB|j,m\rangle &= \frac{j-m}{2j^2}\Big[ \sqrt{(j-m)(j+m+1)}|j,m+1\rangle +\nonumber\\
                    &\qquad\qquad + \sqrt{(j+m)(j-m+1)}|j,m-1\rangle \Big]
            \end{align}
            and thereby 
            \begin{align}
                BAB|j,m\rangle &= \frac{j-m}{2j^2} \Big[ \frac{j-(m+1)}{j} \sqrt{(j-m)(j+m+1)}|j,m+1\rangle + \nonumber\\
                &\qquad \qquad +  \frac{j-(m-1)}{j} \sqrt{(j+m)(j-m+1)}|j,m-1\rangle\Big].
            \end{align}
            We thus conclude that 
            \begin{align}
                \langle j,m|ABAB|j,m\rangle &= \frac{j-m}{4j^3} \Big[ \frac{j-m-1}{j}(j-m)(j+m+1) + 
                \frac{j-m+1}{j}(j+m)(j-m+1) \Big] \\
                    &= \frac{j-m}{4j^4} \left[ (j-m-1)(j-m)(j+m+1) + (j-m+1)^2(j+m) \right].
            \end{align}
            Moving to $m=j-d$, we have $j-m-1=d-1$, $j+m+1=j+(j-d)+1=2j-d+1$, $j+m=j+(j-d)=2j-d$ and thus 
            \begin{align}
                \langle j,j-d|ABAB|j,j-d\rangle = \frac{d}{4j^4} \left[ d(d-1)(2j-d+1) + (d+1)^2(2j-d) \right].
            \end{align}
            By applying the same logic to the other four terms, we obtain the universal upper bound $C\frac{(1+d)^3}{j^3}$ for $X\in\{ABAB, BABA, ABBA,BAAB\}$. 
            Therefore 
            \begin{align}
                \Tr(X\rho_0) &\leq C\cdot \Tr((\eins+D_z)^3J^{-3}\rho_0)\\
                    &\leq C\cdot \sqrt{\Tr((\eins+D_z)^6\rho_0)\Tr(J^{-6}\rho_0)}\\
                    &=\mathcal O(k^{-3}).
            \end{align}
            Here, we used Lemma \ref{lem:bound-on-Dz} and Lemma \ref{lem:bound-for-Jinverse}.
     \end{proof}

\end{document}